\documentclass{CSML}

\def\dOi{11(4:14)2015}
\lmcsheading%
{\dOi}
{1--27}
{}
{}
{Dec.~23, 2014}
{Dec.~22, 2015}
{}

\ACMCCS{[{\bf Information systems}]: Data management systems---Query languages\,/\,
  Database management system engines---Database query processing}
\subjclass{[Database Management]: Systems -- Query processing; Languages -- Query languages.}

\usepackage{amsmath,amsfonts,amsthm,hyperref}
\usepackage{xspace}
\usepackage{tikz}
\usetikzlibrary{decorations.pathreplacing}

\newtheorem{theorem}{Theorem}

\newtheorem{claim}[theorem]{Claim}
\newtheorem{myexample}[theorem]{Example}

\def\ac0{u{\sc AC$^0$}}
\def\sac1{u{\sc SAC$^1$}}
\def\nc1{u{\sc NC$^1$}}
\def\tc0{u{\sc TC$^0$}}

\def\ptime{{\sc PTime}\xspace}

\def\np{{\sc NP}\xspace}
\def\conp{{\sc coNP}\xspace}

\def\expspace{{\sc ExpSpace}\xspace}
\def\twoexpspace{{\sc 2ExpSpace}\xspace}

\newcommand{\pv}{\boldsymbol ;}
\newcommand{\po}{\boldsymbol (}
\newcommand{\pf}{\boldsymbol )}

\newcommand\set[1]{\{#1\}}
\newcommand\size[1]{\ensuremath{|#1|}}

\newcommand{\V}{\textbf{V}\xspace}
\newcommand{\D}{\text{D}\xspace}
\newcommand{\R}{\text{R}\xspace}
\renewcommand{\S}{\text{E}\xspace}
\newcommand{\Q}{\text{Q}\xspace}
\newcommand{\Qlk}{Q_{l,k}}
\renewcommand\path{\text{Path}}
\newcommand{\RPQ}{RPQ\xspace}
\newcommand{\RPQs}{{\RPQ}s\xspace}
\newcommand{\CRPQ}{CRPQ\xspace}

\newcommand{\CQ}{CQ\xspace}
\newcommand{\Datalog}{Datalog\xspace}
\newcommand{\Dataloglk}{Datalog$_{l,k}$\xspace}

\newcommand{\cL}{\mathcal{L}}
\newcommand{\cLR}{\mathcal{L_R}}

\newcommand{\MSO}{{\rm MSO}}

\newcommand{\cert}{cert}
\newcommand{\tempQV}{T_{\Q, \V}}

\newcommand{\f}{\ensuremath{f}\xspace}
\newcommand{\CSP}{\text{CSP}\xspace}

\title[Datalog Rewritings of Regular Path Queries using Views]{Datalog Rewritings of Regular Path Queries \\using Views}

\author[N.~Francis]{Nadime Francis\rsuper a}
\address{{\lsuper a}ENS-Cachan, Inria}
\email{francis@lsv.ens-cachan.fr}

\author[L.~Segoufin]{Luc Segoufin\rsuper b}
\address{{\lsuper b}Inria, ENS-Cachan}
\email{luc.segoufin@inria.fr}

\author[C.~Sirangelo]{Cristina Sirangelo\rsuper c}
\address{{\lsuper c}LSV at ENS-Cachan, Inria, CNRS}
\email{cristina@liafa.univ-paris-diderot.fr}

\keywords{Regular Path Queries, Views, Rewriting, Datalog}

\begin{document}

\begin{abstract}
	We consider query answering using views on graph databases, i.e. databases
	structured as edge-labeled graphs.  We mainly consider views and queries
	specified by Regular Path Queries (\RPQ).  These are queries selecting pairs of nodes
	in a graph database that are connected via a path whose sequence of edge labels
	belongs to some regular language. We say that a view \V determines a query \Q
	if for all graph databases \D, the view image $\V(\D)$ always contains enough
	information to answer \Q on \D. In other words, there is a well defined
	function from $\V(\D)$ to $\Q(\D)$.

	Our main result shows that when this function is monotone, there exists a
	rewriting of \Q as a \Datalog query over the view instance $\V(\D)$. In
	particular the rewriting query can be evaluated in time polynomial in the size
	of $\V(\D)$.  Moreover this implies that it is decidable whether an \RPQ query
	can be rewritten in \Datalog using \RPQ views.
\end{abstract}

\maketitle

\pagestyle{headings}

\section{Introduction}\label{intro}

We consider the problem of answering queries using views on graph
databases. Graph databases are relational databases where all relation symbols
are binary. In other words a graph database can be viewed as an edge-labeled
directed graph.

Graph-structured data can be found in many important scenarios.  Typical examples
are the semantic Web via the format RDF and social
networks. Graph-structured data differs conceptually from relational databases
in that the topology of the underlying graph is as important as the
data it contains. Usual queries will thus test whether two nodes are
connected and \emph{how} they are connected \cite{Barcelo-pods13}.

In many contexts it is useful to know whether a given set of queries can be used
to answer another query. A typical example is the data integration setting where 
data sources are described by views
of a virtual global database. Queries over the global
database are then rewritten as queries over the views.  Another example is
caching: answers to some set of queries against a data source are cached, and
one wishes to know if a newly arrived query can be answered using the cached
information, without accessing the source. This problem also finds application in the
context of security and privacy. Suppose access to some of the information in a
database is provided by a set of public views, but answers to other queries are
to be kept secret.  This requires verifying that the disclosed views do not
provide enough information to answer the secret queries.

All these problems can be phrased in terms of views and query rewriting using
views, which is a typical database problem, not specific to graph databases, 
that has received considerable attention (see~\cite{LevyMSS95, NashSV10, Afrati11} among others). 
When graph databases are concerned, the
difference lies only in the kind of queries under consideration ~\cite{calvanese2000view, calvanese02jcss, calvanese2002lossless,calvanese07tcs}. 

Over graph databases, typical queries have at least the expressive power of
\emph{Regular Path Queries} (\RPQ), defined in~\cite{Cruz-sigmodRec87} (see
also the survey~\cite{Barcelo-pods13}). An \RPQ selects pairs of nodes
connected by a path whose sequence of edge labels satisfies a given regular
expression.  A view, denoted by \V, is then specified using a finite set of
\RPQ{}s. When evaluated over a graph database \D, the view \V yields a new
graph database $\V(\D)$ where each $V_i\in\V$ is a new edge relation symbol.

We are interested in knowing whether the view \V always provides enough
information to answer another \RPQ query \Q, i.e. whether $\Q(\D)$
can be computed from $\V(\D)$ for all databases \D.  When this is the case we say
that \V \emph{determines} \Q, and we look for a
\emph{rewriting} of \Q using \V, i.e. a new query, in some query
language, that expresses $\Q$ in terms of $\V$.
We are then interested in finding an algorithm for evaluating the rewriting, i.e. an algorithm computing $\Q(\D)$ from $\V(\D)$. 

These two related questions, determinacy and query rewriting, have been studied
for relational databases and graph databases. Over relational databases,
determinacy is undecidable already if the queries and views are defined by
union of conjunctive queries, and its decidability status is open for views and
queries specified by conjunctive queries (\CQ)~\cite{NashSV10}. Over graph
databases and \RPQ queries and views, the decidability status of determinacy is
also open~\cite{calvanese2002lossless}.  Determinacy has been shown to be
decidable in a scenario where views and queries can only test whether there is
a path of distance $k$ between the two nodes, for some given
$k$~\cite{Afrati11}.
This scenario lies at the intersection of \CQ and \RPQ and contains already non
trivial examples. For instance the view $\path_3$ and $\path_4$, giving
respectively the pairs of nodes connected by a path of length~3 and~4,
determines the query $\path_5$ asking for the pairs of nodes connected by a
path of length~5~\cite{Afrati11} (see also Example~\ref{example-determinacy} in
Section~\ref{section-prelim}).

Clearly when \Q can be rewritten in terms of \V, the rewriting witnesses that
\V determines \Q. On the other hand determinacy does not say that one can find
a rewriting definable in a particular language, nor with particular computational
properties.

It is then natural to ask which rewriting language $\cLR$ is sufficiently
powerful so that determinacy is equivalent to the existence of a rewriting
definable in $\cLR$. This clearly depends on the language used for defining the
query and the view.

Consider again the case of $\path_5$ that is determined by $\path_3$ and
$\path_4$. A rewriting  $R(x,y)$ of $\path_5$ in terms of $\path_3$ and $\path_4$ is
defined by:
\begin{equation*}
\exists u ~ ( \path_4(x,u) \land \forall v ~(\path_3(v,u) \rightarrow
\path_4(v,y)) )
\end{equation*}
and it can be shown that there is no rewriting definable in \CQ, nor in \RPQ
(cf. Example~\ref{example-determinacy}).  In the case of views and queries
defined by \CQ{}s it is still an open problem to know whether first-order logic
is a sufficiently powerful rewriting language.  Even worse, it is not even
known whether there always exists a rewriting that can be evaluated in time
polynomial in the size of the view instance~\cite{NashSV10}, ie. polynomial \emph{data
  complexity}. A similar situation arises over graph-databases and \RPQ{} views
and queries~\cite{calvanese2002lossless}.

\medskip

It can be checked that in the example above there exists no monotone rewriting
of $\path_5$ (see again Example~\ref{example-determinacy}). In particular, as \RPQ{}s define only monotone
queries, no rewriting is definable in \RPQ.  Monotone query languages such as
\CQ, \Datalog, \RPQ and their extensions are of crucial importance in many
database applications.  The possibility of expressing rewritings in these
languages is subject to a monotonicity restriction.

This is why in this paper we are considering a stronger notion of determinacy,
referred to as \emph{monotone determinacy}, by further requiring that the
mapping from view instances to query results is \emph{monotone}.

In the case when views and queries are defined by \CQ{}s, monotone determinacy
can be shown to be equivalent to the existence of a rewriting in
\CQ\cite{NashSV10}. As this latter problem is decidable~\cite{LevyMSS95},
monotone determinacy for \CQ{}s is decidable.

We consider here monotone determinacy for graph databases and views and queries
defined by \RPQ{}s.
We first observe that monotone determinacy corresponds to the notion called
\emph{losslessness under the sound view assumption}
in~\cite{calvanese2002lossless}, where it was shown to be decidable. We then
concentrate on the rewriting problem.

We know that there exist cases of monotone rewritings that are not expressible
in \RPQ~\cite{calvanese2002lossless} (see also Example~\ref{example-norpq} in
Section~\ref{section-monotone}). We thus need a more powerful language in order
to express all monotone rewritings.

It is not too hard to show that if \V determines \Q then there exists a
rewriting with \np data complexity, as well as a rewriting
with \conp data complexity.  Our main result shows that if moreover \V determines \Q
in a monotone way, there exists a rewriting definable in \Datalog, which therefore can
 be evaluated in polynomial time.

Our proofs are constructive, hence the \Datalog rewriting can be computed from
\V and \Q. 

As a corollary this implies that it is decidable whether a query \Q has a
rewriting definable in \Datalog using a view \V, where both \V and \Q are
defined using \RPQ{}s. This comes from the fact that our main result implies
that the existence of a rewriting in \Datalog is equivalent to monotone
determinacy, a decidable property as mentioned above.

\medskip

\paragraph{Related work}
The work which is most closely related to ours is that of the ``Four
Italians''.  In particular, the notion of losslessness under the exact view
assumption introduced in~\cite{calvanese2002lossless} corresponds to what we
call determinacy; similarly the notion of losslessness under the sound view
assumption corresponds to what we call monotone determinacy. Monotone determinacy is
also mentioned in the thesis~\cite{perez} under the name of ``strong
determinacy''. It is shown there that it corresponds to the existence of a
monotone rewriting.

A lot of attention has been devoted to the problem of computing the set of
certain answers to a query w.r.t a set of views, under the sound view
assumption (see the precise definition of certain answers in
Section~\ref{section-certain}). For \RPQ views and queries, the problem is
shown to be equivalent to testing whether the given instance homomorphically
embeds into a structure $\tempQV$ computed from the view \V and the query
\Q~\cite{calvanese2000view}. In general this shows that the data complexity of
computing the certain answers is \conp-complete. Building on results on
Constraint Satisfaction Problems~\cite{feder1998computational}, it was also
shown in~\cite{calvanese2000view} that for an \RPQ view \V, an \RPQ query \Q
and for each $l,k$, with $l \leq k$, there is a \Datalog program $\Qlk$ which
is contained in the certain answers to \Q given \V and is, in a sense,
maximally contained: i.e. $\Qlk$ contains all \Datalog programs which are
contained in the certain answers and use at most $l$ head variables and at most
$k$ variables in each rule.

If we assume that \V determines \Q in a monotone way, it is easy to see that
the query computing the certain answers under the sound view assumption is a
rewriting of \Q using \V (i.e the certain answers of a view instance
$\V(\D)$ are precisely the query result $\Q(\D)$).

However there are possibly other rewritings (they only need to agree on
instances of the form $\V(\D)$, but may possibly differ on instances not in the
image of \V.)  While computing the certain answers is \conp-hard, our
main result shows that there exists another rewriting which is expressible in
\Datalog, and has therefore polynomial time data complexity.

Nevertheless our proof makes use of the structure $\tempQV$ mentioned above, and
our \Datalog rewriting turns out to be the query $\Qlk$ associated with \Q and
\V for some suitable values of $l$ and $k$.

\section{Preliminaries}\label{section-prelim}
\paragraph{Graph databases and paths}
A binary schema is a finite set of relation symbols of arity 2. All the schemas
used in this paper are binary. A \emph{graph database} \D is a finite
relational structure over a (binary) schema $\sigma$. We will also say a
$\sigma$-structure. Alternatively \D can be viewed as a directed edge-labeled
graph with labels from the alphabet $\sigma$. The elements of the domain of \D
are referred to as \emph{nodes}. The number of elements in \D is denoted by
$\size{\D}$. If $A$ is a set of elements of \D, we denote by $\D[A]$ the
substructure of \D induced by $A$.

Given a graph database \D, a \emph{path} $\pi$ in \D from $x_0$ to $x_m$ is a
finite sequence $\pi = x_0a_0x_1\ldots x_{m-1}a_{m-1}x_m$, where each $x_i$ is
a node of \D, each $a_i$ is in $\sigma$, and $a_i(x_i,x_{i+1})$ holds in \D
for each $i$. 
A \emph{simple path} is a path
such that no node occurs twice in the sequence.
The \emph{label} of $\pi$, denoted by $\lambda(\pi)$, is the word $a_0a_1\ldots
a_{m-1} \in \sigma^*$. By abuse of notation, we sometimes view a path $\pi$ as a
graph database, which contains only the nodes and edges that occur in the
sequence.
\paragraph{Queries and query languages}
  A \emph{query} \Q over a schema $\sigma$ is a mapping
associating to each graph database \D over $\sigma$ a finite relation
$\Q(\D)$ over the domain of \D. We will only consider binary queries, that is queries
that return binary relations, and work with the following query languages.

A \emph{Regular Path Query} (often abbreviated as \RPQ) over $\sigma$ is given by a regular
expression over the alphabet $\sigma$. If \Q is an \RPQ, we denote by $L(\Q)$ the language
corresponding to its regular expression. On a graph database, such a query
selects all the pairs $(x, y)$ of nodes such that there exists a path $\pi$
from $x$ to $y$ with $\lambda(\pi) \in L(\Q)$.
For instance the query $\path_3$ of the introduction is an \RPQ corresponding
to the regular expression $\sigma\sigma\sigma$ (also denoted
$\sigma^3$). Another example is the \RPQ $(\sigma\sigma)^*$ that select pairs
of nodes connected via a path of even length.

A \emph{Context-Free Path Query}
over $\sigma$
is defined similarly but using a context-free grammar instead of a regular
expression.

A \emph{Conjunctive Regular Path Query} (sometimes abbreviated \CRPQ) over $\sigma$ is a conjunctive
query whose atoms are specified using \RPQ{}s over $\sigma$.
For instance the query
\begin{equation*}
  \exists z~ Q_1(x,z) \land Q_2(z,y) \land Q_3(z,y)
\end{equation*}
where $Q_1 = a^+$, $Q_2 = b$ and $Q_3 = c$ selects pairs of nodes $(x,y)$ which
are connected via a path labeled $a^+b$ and another path labeled $a^+c$ sharing
their $a^+$ part. This cannot be expressed by an \RPQ.

A \emph{\Datalog query} over schema $\sigma$ is defined by a finite set of rules of the form
\begin{equation*}
  I(\bar x) :\!\!-~ I_1(\bar x_1) \land \cdots \land I_m(\bar x_m)
\end{equation*}
where each $I_i$ is a relational symbol, either a symbol from $\sigma$, or an
\emph{internal symbol}. $I(\bar x)$ is called the \emph{head} of the rule and $I$ must
be an internal symbol. The variables $\bar x$ are among $\bar x_1 \dots \bar x_m$ and
the variables of $\bar x_i$ not occurring in $\bar x$
should be understood as existentially quantified. One of the internal symbols,
referred to as the \emph{goal},
is binary and is designated as being the output of the query. The evaluation of
a \Datalog query computes the internal relations incrementally
starting from the empty ones by applying greedily the rules (see~\cite{AHV95}).

It is easy to see that any Regular or Context-Free Path Query, and
therefore any Conjunctive Regular Path Query, can be expressed in \Datalog.  Hence
\Datalog is the most expressive of the query languages presented above. It is
also well known that each \Datalog query can be evaluated in polynomial time,
data complexity, using the procedure briefly sketched above.

We will consider restrictions of \Datalog limiting the maximal arity of
the internal symbols and the number of variables in each rule. This is
classical in the context of Constraint Satisfaction Problems
(\CSP)~\cite{feder1998computational} that we will use in
Section~\ref{section-datalog}. In the context of \CSP,  \Datalog
programs are boolean (i.e. the goal has arity 0) and 
\Dataloglk denotes the fragment allowing at most $k$ variables in each 
rule and internal symbols of arity at most $l$.
Here we are dealing with \emph{binary}
\Datalog programs. In order to stay close to the
notations and results coming from \CSP, we generalize this definition and let 
\Dataloglk denote  the \Datalog programs having at most $k+r$ variables in each 
rule and internal symbols of arity at most $l+r$, where $r$ is the arity 
of the goal, in our case $r=2$.

\paragraph{Views}
If $\sigma$ and $\tau$ are (binary) schemas, a \emph{view} \V from $\sigma$ to
$\tau$ is a set consisting of one binary query over $\sigma$ for each symbol in
$\tau$.  If $\V$ consists of the queries $\{V_1,\ldots,V_n\}$, with a little
abuse of notation, we let each $V_i$ also denote the corresponding symbol in
$\tau$.  For a graph database \D over $\sigma$, we denote by $\V(\D)$ the graph
database over $\tau$ where each binary symbol $V_i$ is instantiated as
$V_i(\D)$.  We say that a view consisting of the queries $\{V_1,\ldots,V_n\}$
is an \RPQ view if each $V_i$ is an \RPQ. We define similarly Context-Free Path
Query views and Conjunctive Regular Path Query views.

In what follows whenever we refer to a view \V and a query \Q, unless
otherwise specified, we always assume that \Q is over the schema $\sigma$ and \V
is a view from~$\sigma$ to~$\tau$. A \emph{view instance} $\S$ is a
$\tau$-structure such that $\S=\V(\D)$ for some database \D.

\paragraph{Determinacy and rewriting}
The notion of determinacy specifies when a query can be answered completely
from the available view. The following definitions are taken from~\cite{NashSV10}.

\begin{defi}[Determinacy]
	We say that a view \V determines a query \Q
 if :
\begin{equation*}
\forall \D,\D', \quad \V(\D) = \V(\D') \ \Rightarrow \ \Q(\D) = \Q(\D')
\end{equation*}
\end{defi}\smallskip

\noindent In other words, $\Q(\D)$ only depends on the view instance $\V(\D)$ and not on
the particular database \D yielding the view.  Observe that determinacy says
that there exists a function \f defined on view images such that
$\Q(\D)=\f(\V(\D))$ for each database \D. We call \f the \emph{function induced
  by $\Q$ using $\V$}.

A \emph{rewriting} of \Q using \V is a query \R over the schema $\tau$ such that
$\R(\V(\D))=\Q(\D)$ for all~$\D$. 

Notice that there can be possibly many rewritings, while the
function induced by $\Q$ using $\V$ is unique.
In fact the domain of \f is defined to be
the set of view images, that is, all the $\tau$-structures $\S$ such that there
exists a database $\D$ with $\V(\D) = \S$. Thus, \f is fully defined by the
identity $\Q(\D)=\f(\V(\D))$, and is therefore unique. On the contrary, 
rewritings are defined as queries over $\tau$, which means that they are mappings defined
over all $\tau$-structures $\S$, even those which are not of the form 
$\S = \V(\D)$. In particular, this means that the condition
$\Q(\D) = \R(\V(\D))$ is not sufficient to fully define $\R$, as it
can take arbitrary values on $\tau$-structures that are not of the form 
$\V(\D)$. Of course any rewriting coincides with the function \f 
\emph{when restricted to view images}.

\begin{myexample}\label{example-determinacy}
 Consider again the view \V defined by the two \RPQ{}s $V_1=\sigma^3$ and
 $V_2=\sigma^4$ testing for the existence of a path of length 3 and 4, respectively.
 Let $\Q=\sigma^5$ be the \RPQ testing for the existence of a path of
 length 5.

 It turns out that \V determines \Q~\cite{Afrati11}. This is not immediate to
 see but, as mentioned in Section~\ref{intro}, one can verify that a rewriting of 
 \Q using \V can be expressed in first-order by the following query:
 \begin{equation*}
\exists u ~ ( V_2(x,u) \land \forall v ~(V_1(v,u)
\Rightarrow V_2(v,y)))
\end{equation*}
As shown in Figure~\ref{determinacy-figure}, the function induced by \Q using
\V is not monotone. This implies that no monotone query can be a
rewriting, in particular there exists no \CQ nor \RPQ rewriting. 

 Consider now the \RPQ $\Q'=\sigma^2$. One can verify that \V does not
 determine $\Q'$. Indeed the database consisting of a single node with no
 edge, and the database consisting of a single path of length 2, have
 the same empty view but disagree on $\Q'$.
\end{myexample}

\begin{figure}[h!]
	\centering
	\begin{tikzpicture}[shorten >=1pt,->]
		\tikzstyle{vertex}=[circle,fill=black,minimum size=3pt,inner sep=0pt]
		\foreach \name/\x in {0/0, 1/1, 2/2, 3/3, 4/4, 5/5} \node[vertex] (\name) at (\x,5) {};
		\foreach \from/\to in {0/1, 1/2, 2/3, 3/4, 4/5} \draw (\from) -- (\to);
		\foreach \name/\x in {0/0.25, 1/1.25, 2/2.25, 3/3.25, 4/4.25, 5/5.25} \node (lab\name) at (\x, 4.75) {$x_\name$};
		\node (D) at (-0.5,5) {$\D :$};

		\foreach \name/\x in {0/0, 1/1, 2/2, 3/3, 4/4} \node[vertex] (\name) at (\x,0) {};
		\foreach \from/\to in {0/1, 1/2, 2/3, 3/4} \draw (\from) -- (\to);
		\foreach \name/\y in {01/3, 11/2, 21/1} \node[vertex] (\name) at (4,\y) {};
		\foreach \from/\to in {01/11, 11/21, 21/4} \draw (\from) -- (\to);
		\foreach \name/\x in {02/3, 12/2, 13/1} \node[vertex] (\name) at (\x,2) {};
		\foreach \from/\to in {11/02, 02/12, 12/13} \draw (\from) -- (\to);
		\foreach \name/\x/\y in {0/0.25/-0.25, 1/4.25/2.75, 2/4.25/1.75, 3/3.25/-0.25, 4/4.25/-0.25, 5/1.25/1.75} 
		\node (lab\name) at (\x, \y) {$x_\name$};
		\node (D') at (-0.5,2.5) {$\D' :$};
	\end{tikzpicture}
	\caption{Illustration for Example~\ref{example-determinacy}: $\D$ and
          $\D'$ are such that $\V(\D) \subseteq \V(\D')$, but $(x_0,x_5) \in
          \Q(\D)$, whereas $(x_0,x_5) \notin \Q(\D')$. Hence the function
          induced by \Q using \V is not monotone.}
	\label{determinacy-figure}
\end{figure}
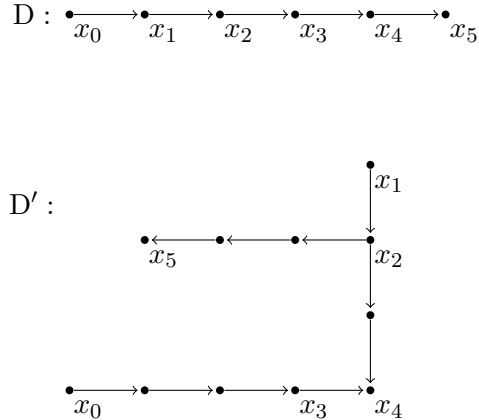

It is important at this point to understand the difference between determinacy
and rewriting. If \V determines \Q then there exists a rewriting of \Q using \V.
However there are possibly many rewritings of \Q using \V. Each of them agrees on
the function induced by \Q using \V when restricted to view images, but
can take arbitrary values on structures that are not in the image of the view.
Consider for instance the view \V and the query \Q of
Example~\ref{example-determinacy}. The query:
 \begin{equation*}
\exists u,u' ~ V_2(x,u) \land V_1(x,u')\land \forall v ~(V_1(v,u)
\Rightarrow V_2(v,y))
\end{equation*}
is also a rewriting of \Q using \V. It is equivalent to the rewriting of
Example~\ref{example-determinacy} on $\tau$-instances $\S$ such that $\S=\V(\D)$
for some \D. Indeed whenever $V_2(x,u)$ holds in $\V(\D)$, the database $\D$ contains a path 
of length 4 from $x$ to $u$, hence if $u'$ is the node at distance 3
from $x$ in this path,  $V_1(x,u')$ also holds in $\V(\D)$. However the two rewritings may
 differ on  instances which are not in the view image, such as an instance consisting of a single 
 $V_2$-labeled edge.

The \emph{determinacy problem} for a query language $\cL$ is the problem of
deciding, given an input view \V defined in $\cL$ and a query \Q of $\cL$, whether \V
determines \Q. 

Determinacy does not say whether there exists a rewriting definable in a
particular query language, or computable with a particular data complexity. This clearly
depends on the language used for specifying the views and queries.

The \emph{rewriting problem} for a query language $\cL$ is the problem of
finding a rewriting for a query \Q of  $\cL$ using a view \V defined in $\cL$
whenever \V determines \Q.

These two problems have been thoroughly investigated in the case that $\cL$ is
\RPQ~\cite{calvanese2002lossless, calvanese2000view,
  calvanese02jcss,calvanese07tcs}. However the determinacy problem for \RPQ
remains wide open and it is not clear what would be a good (low data
complexity) rewriting language for \RPQ. Note that a similar situation arises
in the case that $\cL$ is \CQ~\cite{NashSV10, Afrati11}.

\section{Determinacy problem}

We have already mentioned above that the determinacy problem for \RPQs is open.
For Context-Free Path Queries and for Conjunctive Regular Path Queries,
determinacy is undecidable. Actually the problem is already undecidable when
the query \Q is an \RPQ.  These undecidability results are formalized in the two
following propositions.

\begin{prop}\label{prop-cfpq-determinacy}
  Given a Context-Free Path Query view $\V$ and a Regular Path Query $\Q$,
  it is undecidable whether $\V$ determines $\Q$.
\end{prop}

\begin{proof}
  We prove this by reduction from the universality problem for context-free
  languages. Let $L$ be a context-free language over some alphabet
  $\sigma$. Let $\$$ be a fresh symbol that does not appear in $\sigma$. Let
  $\V = \set{V}$ where $V$ is defined by $L(V) = \$\cdot L \cdot \$$. Let $\Q$
  be defined by $L(\Q) = \$\cdot \sigma^* \cdot \$$. Then $\V$ determines $\Q$
  if and only if $L$ is universal over $\sigma$.

	\begin{itemize}
        \item Assume that $L$ is universal. Then $\Q = V$ and it is easy to
          check that $R = V$ is a rewriting of $\Q$ using $\V$.
        \item Conversely, assume that $L$ is not universal. Then there exists
          $w \in \sigma^*$ such that $w \notin L$. Consider the database $\D$
          consisting of a simple path labeled by $\$\cdot w \cdot\$$, and the
					empty database $\D'$. Then
          $\V(\D) = \emptyset = \V(\D')$, but $\Q(\D)$ contains the first and 
					last node of the path, whereas $\Q(\D')$ is empty. 
					Hence, $\V$ does not determine $\Q$.
	\end{itemize}
\end{proof}

\begin{prop}
  Given a Conjunctive Regular Path Query view $\V$ and a Regular Path Query
  $\Q$, it is undecidable whether $\V$ determines $\Q$.
\end{prop}

\begin{proof}
  We prove this by reduction from the word problem for graph databases.

\begin{center}
\fbox{
$
	\begin{array}{ll}
          \textsc{Problem} : & \textsc{Word problem for graph databases} \\
          \textsc{Input} : & \text{A list of pairs } (u_i,v_i)_{0 < i \leq
            n}, \text{a pair } (u,v), 
             \text{ where } u, v, u_i, v_i, \text{ for every } i \\
            & \text{are words over } \sigma, \text{ viewed as \RPQ{}s } \\
          \textsc{Question} : & \text{Is the following statement true?} \\
          & \text{For every graph database \D  over } \sigma, \text{if } \forall i, u_i(\D) = v_i(\D), \text{ then } u(\D) = v(\D) 
	\end{array}
$}
\end{center}

A straightforward reduction from the word problem for finite semigroups shows:

\begin{lem}
  The word problem for graph databases is undecidable.
\end{lem}

\begin{proof}
  We prove this by reduction from the word problem for finite semigroups. This
  problem has the same input as the word problem for graph databases but asks
  whether for all semigroup $S$ and all homomorphism $h$ from $\sigma^*$ to $S$
  such that $h(u_i)=h(v_i)$ for all $i$, it is the case that $h(u)=h(v)$.
 
  We now prove that any input is accepting for the word problem for finite
  semigroups if and only if it is accepting for the word problem for graph databases.

  \begin{enumerate}
  \item Assume that the input is accepting for the word problem for finite
    semigroups. Let \D be a graph database such that for all $i$, $u_i(\D) =
    v_i(\D)$. From $\D$, we compute the semigroup $S_\D$ and the homomorphism $h :
    \sigma^* \rightarrow S_\D$ as follows:
    \begin{itemize}
    \item The elements of $S_\D$ are the set of pairs $w(\D)$ for all $w\in
      \sigma^*$. As \D is finite $S_\D$ is finite.
    \item Let $x$ and $y$ be two elements of $S_\D$. Let $u,v \in \sigma^*$ such
      that $x = u(\D)$ and $y = v(\D)$. Then $x\cdot y$ is defined as $u\cdot
      v(\D)$. It is easy to check that this operation is associative and well
      defined (i.e. does not depend on the specific choice of $u$ and $v$).
    \item For all $\alpha \in \sigma$ we set $h(\alpha) = \alpha(\D)$. Hence
      for all $u\in\sigma^*$ we have $h(u)=u(D)$.
    \end{itemize}
    By construction we therefore have for all $i$, $h(u_i) =
    h(v_i)$. Hence, $h(u) = h(v)$, which implies that $u(\D) = v(\D)$.

  \item Assume that the input is accepting for the word problem for
    graph databases. Let $S$ be a finite semigroup, and $h$ an homomorphism from
    $\sigma^*$ to $S$, such that, for all $i$, $h(u_i) = h(v_i)$. From $S$ and
    $h$, we define the graph database $\D_h$ as follows:
    \begin{itemize}
    \item The sets of nodes of $\D_h$ is $h(\sigma^+) \cup
      \set{\varepsilon}$. This set is finite since $h(\sigma^+)$ is a subset of
      $S$.
    \item Let $x$ and $y$ be two nodes of $\D_h$. Then there is an edge $\alpha$
      from $x$ to $y$ if either $x = \varepsilon$ and $y = h(\alpha)$ or $x
      \neq \varepsilon$ and $x \cdot h(\alpha) = y$.
    \end{itemize}
    Assume that $(x,y) \in u_i(\D_h)$. Then either $x = \varepsilon$, hence $y
    = h(u_i) = h(v_i)$ and $(x,y)\in v_i(\D_h)$, or $x \cdot h(u_i) = y$, which
    implies that $x \cdot h(v_i) = y$ and $(x,y)\in v_i(\D_h)$. Hence,
    $u_i(\D_h)=v_i(\D_h)$ for all $i$ and therefore $u(\D_h) = v(\D_h)$. Hence,
    $(\varepsilon,h(u)) \in v(\D_h)$, which implies that there is a path $v$ from
    $\varepsilon$ to $h(u)$ and thus that $h(u) = h(v)$.\qedhere
  \end{enumerate}
\end{proof}

 \noindent Let
  $(u_i,v_i)_{0 < i \leq n}$ and $(u,v)$ be an input
  for the word problem. Let $\sigma'$ be a copy of $\sigma$ using only fresh
  symbols.  For each $\alpha\in\sigma$, we use $\alpha'$ to denote the
  corresponding symbol in $\sigma'$. We define the following query and view:
  \begin{itemize}
  \item $\Q$ is the \RPQ defined by $L(\Q) = \set{u,v'}$ where $v'$ is a copy
    of $v$ using symbols of $\sigma'$.
  \item For all $\alpha \in \sigma$, $V_\alpha$ is a query of the view defined
    by the \RPQ $L_\alpha = \set{\alpha,\alpha'}$.
  \item For all $i$, $V_i$ is also a query of the view defined by the \RPQ $L_i = \set{u_i,v'_i}$, where
    $v'_i$ is a copy of $v_i$ using symbols of $\sigma'$.
  \item For all $\alpha,\beta \in \sigma$, $T_{\alpha,\beta}$ is a query of the view defined
    by the \CRPQ: $\alpha(x,y) \wedge \exists z,t ~\beta'(z,t)$.
  \item For all $\alpha,\beta \in \sigma$, $T'_{\alpha,\beta}$ is a query of the view defined
    by the \CRPQ: $\alpha'(x,y) \wedge \exists z,t ~ \beta(z,t)$.
  \end{itemize}

  We now prove that $\V = \set{V_\alpha, V_i, T_{\alpha,\beta}, T'_{\alpha,\beta} \ | \
    \alpha,\beta\in\sigma, 0< i \leq n}$ determines $\Q$ if and only if the input is
  accepting for the word problem for graph databases.

  \begin{enumerate}
  \item Assume that the input is accepting for the word problem for graph
    databases. Let $\D$ and $\D'$ be two graph databases such that $\V(\D) =
    \V(\D')$. Consider first the case where $\D$ uses symbols from both
    $\sigma$ and $\sigma'$, then $T_{\alpha,\beta}$ and $T'_{\alpha,\beta}$
    reveal $\D$ entirely, which implies that $\D = \D'$, and thus $\Q(\D) =
    \Q(\D')$. Similarly, if both $\D$ and $\D'$ use only symbols from $\sigma$,
    then $V_\alpha$ reveals $\D$ entirely ensuring that $\D = \D'$. It remains
    to consider the case where $\D$ only uses symbols from $\sigma$ and $\D'$
    only uses symbols from $\sigma'$. Notice that, since $V_\alpha(\D) = V_\alpha(\D')$, then $\D$ and $\D'$
    are isomorphic (by renaming each $\alpha$ to $\alpha'$).

    Let $(x,y) \in u_i(\D)$. Hence, $(x,y) \in V_i(\D)$, which implies that
    $(x,y)\in V_i(\D')$, and finally that $(x,y)\in v'_i(\D')$. By isomorphism
    $(x,y)\in v_i(\D)$. Similarly, we can show that $(x,y)\in v_i(\D)$ implies
    $(x,y) \in u_i(\D)$. Hence, $u(\D) = v(\D)$. Let $(x,y) \in \Q(\D)$. Then,
    $(x,y) \in u(\D)$, which implies that $(x,y) \in v'(\D')$, and thus that
    $(x,y) \in \Q(\D')$. A similar reasoning also gives the converse, and we
    can conclude that $\V$ determines $\Q$.

  \item Assume that $\V$ determines $\Q$. Let $\D$ be a graph database over
    $\sigma$ that satisfies the condition for the word problem.  Let $\D'$ be
    the copy of $\D$ given by renaming the symbols in $\sigma$ by the
    corresponding symbols in $\sigma'$. Remark now that $\V(\D) =
    \V(\D')$. Indeed, $V_\alpha(\D) = V_\alpha(\D')$ is given by the fact that
    $\D'$ is a copy of $\D$ over $\sigma'$.  $V_i(\D) = V_i(\D')$ is given by
    the fact that $\D$ satisfies the condition for the word problem. Finally,
    $T_{\alpha,\beta}(\D) = T_{\alpha,\beta}(\D') = T'_{\alpha,\beta}(\D) = T'_{\alpha,\beta}(\D') = \emptyset$
    comes from the fact that $\D$ (resp. $\D'$) uses only symbols from $\sigma$
    (resp. $\sigma'$).

    Since $\V$ determines $\Q$, this implies that $\Q(\D) = \Q(\D')$. Let
    $(x,y) \in u(\D)$. Then $(x,y) \in \Q(\D)$, which implies that $(x,y) \in
    \Q(\D')$. Hence, $(x,y) \in v'(\D')$, and since $\D'$ is a copy of $\D$,
    this yields $(x,y) \in v(\D)$. A similar reasoning also gives the converse,
    and we can conclude that the input is accepting for the word problem for
    graph databases.
  \end{enumerate}
\end{proof}

\section{Views and Rewriting}
We have seen in the previous section that knowing whether a given view \V
determines a given query \Q is often computationally a difficult task. In this
section we assume that \V determines \Q and we investigate how \Q can be computed
from the given view instance. 

A possibility is to use the following generic algorithm :
\medskip

\noindent\emph{Given a $\tau$-structure \S, compute a $\sigma$-structure \D such that
$\V(\D)=\S$ (reject if no such \D exists) and return $\Q(\D)$.}
\medskip

As we know that
\V determines \Q this procedure always returns the correct answers on view images. Therefore the query over $\tau$ defined by this algorithm is a rewriting of \Q using \V.

For all the query languages considered in this paper, computing $\V(\D)$ and
$\Q(\D)$ can be done in time polynomial in \size{\D}. Hence it remains to be
able to test whether there exists a \D such that $\V(\D)=\S$ and, if yes,
compute such a \D.

The first issue, testing whether a $\tau$-instance is in the image of the view, is already
a challenging task and will be investigated in the next section. We start with
the second problem, i.e. computing a \D such that $\V(\D)=\S$, if it exists.

\subsection{Looking for a view preimage}

We assume in this section that \V is a view from $\sigma$ to $\tau$ and that
we are given a $\tau$-structure \S that is in the image of \V. We are now
looking for a \D such that $\V(\D)=\S$, knowing that one such \D exists.
Our first result below shows that for \RPQ views, if such a \D exists then
there is one whose size is polynomial in \size{\S}. It is essentially a pumping
argument.

\begin{lem}\label{lemma-inverse-view}
  Let \V be an \RPQ view from $\sigma$ to $\tau$. Let \S be a
  $\tau$-structure.  If $\S=\V(\D)$ for some \D 
  then $\S=\V(\D')$, for some $\D'$ of size quadratic in \size{\S}.
\end{lem}

\begin{proof}
  Let \V and \S be as in the statement of the lemma.  We show that if there
  exists $\D$ such that $\S=\V(\D)$ then there exists a new database $\D'$ of
  size $O(\size{\S}^2)$ such that $\V(\D') = \V(\D)$.  $\D'$ is obtained from
  $\D$ in several steps.  First $\D$ is ``normalized", without altering its
  view, so that nodes not occurring in \S appear in only one path linking two
  nodes of \S.  The normalized $\D$ turns out to consist of a constant number
  of disjoint paths between each pair of nodes of $\S$ (where the constant only
  depends on the size of the view automaton).  Then a Ramsey argument is used
  to show that these paths can be ``cut" without changing the view.  The
  resulting database $\D'$ thus consists of a constant number of paths of
  constant length between each pair of nodes of $\S$. The size of $\D'$ is
  therefore $O(\size{\S}^2)$. We now formalize this argument.

  Assume that there exists a database \D such that $\S=\V(\D)$. We prove the
  lemma by constructing a new database $\D'$ such that $\V(\D') = \V(\D)$, with
  $\size{\D'} = O(\size{\S}^2)$.

  Let $A=\langle S_\V,\delta_\V,q^0_\V,F_\V \rangle$ be the product automaton
  of all the deterministic minimal automata of all the regular expressions of
  the \RPQ{}s in \V. Let $N(\V)$ be the number of states of $A$, i.e $|S_\V|$.

  In what follows, for $w \in \sigma^*$, $\delta_\V(\cdot,w)$ denotes the
  function from $S_\V$ to $S_\V$ sending $q$ to $p$ such that there is a run of
  $A$ on $w$ starting in state $q$ and arriving in state~$p$.

  We say that a path $\pi$ from $u$ to $v$ in a database $\D'$ is $\V$-minimal
  if $u,v$ are elements of $\V(\D')$ and no other nodes of $\pi$ are in the domain of $\V(\D')$.
        
  We first build a database $\D_1$ such that :

  \begin{itemize}
  \item $\V(\D_1)=\V(\D)$;
  \item each node of $\D_1$ is in a $\V$-minimal path and no two $\V$-minimal
    paths in $\D_1$ intersect;
  \item the number of $\V$-minimal paths in $\D_1$ is bounded by
    $\size{\V(\D)}^2 \cdot N(\V)^{N(\V)}$.
  \end{itemize}

  $\D_1$ is constructed as follows: All elements of $\V(\D)$ are elements of
  $\D_1$. Moreover, for each function $f : S_\V \rightarrow S_\V$ and each pair
  $(x,y)$ of elements of $\V(\D)$, if there exists a $\V$-minimal path $\pi$
  from $x$ to $y$ in \D and such that $f = \delta_\V(\cdot,\lambda(\pi))$, then
  we add to $\D_1$ a copy of $\pi$ that uses only fresh, non-repeating nodes,
  except for $x$ and $y$.  Figure \ref{figure-inverse-view} illustrates the
  main idea of this construction.

  It is now easy to check that $\D_1$ has the desired properties. The second
  bullet holds by construction. Clearly the number of $f: S_\V \rightarrow
  S_\V$ is bounded by $N(\V)^{N(\V)}$ hence the third bullet holds. It remains
  to check that $\V(\D_1) = \V(\D)$. There is an obvious canonical homomorphism
  sending $\D_1$ to $\D$. Hence $\V(\D_1) \subseteq \V(\D)$. For the converse
  direction, consider a path $\pi$ witnessing the fact that
  $(u,v)\in\V(\D)$. Decompose $\pi$ into $\V$-minimal paths. By construction,
  each of these $\V$-minimal paths can be simulated in $\D_1$. Hence $(u,v) \in
  \V(\D_1)$.

  \medskip

  From $\D_1$ we construct the desired $\D'$ by replacing each 
  $\V$-minimal path of $\D_1$ by another one whose length is
  bounded by a constant $r$ and without affecting the view image. Altogether
  $\D'$ will have a size bounded by $r \cdot \size{\V(\D)}^2 \cdot N(\V)^{N(\V)}$,
  hence polynomial in $|\V(\D)|$ as desired.

  Let $r$ be the Ramsey's number that guarantees the existence of a
  monochromatic 3-clique in an $r$-clique using $N(\V)^{N(\V)}\cdot
  2^{N(V)^{N(V)}}$ colors.

  Consider a $\V$-minimal path $\pi = xa_0x_1a_1\ldots x_{m}a_{m}y$ in
  $\D_1$ such that $m > r$. For $1\leq s<t\leq m$ we denote by $\pi_{s \rightarrow t}$
  the subpath of $\pi$ that starts at position $s$ and ends at position $t$, that is
  $\pi_{s \rightarrow t} = x_sa_sx_{s+1}a_{s+1}\ldots a_{t-1}x_t$.

  To each pair of nodes $(x_i,x_j)$ in $\pi$ with $i<j$, we attribute the color
  $(f_{ij},\Delta_{ij})$ where:
\begin{align*}
  f_{ij} &= \delta_\V(\cdot,\lambda(\pi_{i \rightarrow j}))\\
  \Delta_{ij} &= \set{f : S_\V \rightarrow S_\V \ | \ \exists \alpha, &&i<\alpha<j
  \textrm{ and }\\
&&& f = \delta_\V(\cdot,\lambda(\pi_{i \rightarrow
    \alpha}))}.
\end{align*}
Then, by our choice of $r$, we know that there exist $i<j<k$ such that
$f_{ij}=f_{jk}=f_{ik}$ and $\Delta_{ij}=\Delta_{jk}=\Delta_{ik}$. Let $\pi'$ be
the path constructed from $\pi$ by replacing the
subpath $\pi_{i \rightarrow k}$ by $\pi_{j \rightarrow k}$.

Let $\D_2$ be the database constructed from $\D_1$ by replacing $\pi$ by
$\pi'$. We now prove that $\V(\D_2) = \V(\D_1)$. As $\D_2$ still has all the
properties of $\D_1$ listed above, by repeating this operation until all
$\V$-minimal paths have length less than $r$ we eventually get the desired
database $\D'$.
	
Let $(u,v) \in \V(\D_1)$ as witnessed by a path $\mu$ in $\D_1$. Then $\mu$
neither starts nor ends in an internal node of $\pi$ as internal nodes do not appear in
$\V(\D_1)$. Hence either $\mu$ does not use $\pi$ or it uses all of it.  In the
former case, $\mu$ witnesses the fact that $(u,v)\in\V(\D_2)$. In the latter,
notice that $f_{ik}=f_{jk}$ implies that $\lambda_\V(\cdot,\lambda(\pi)) =
\lambda_\V(\cdot,\lambda(\pi'))$, hence replacing $\pi$ by
$\pi'$ in $\mu$ witnesses the fact that $(u,v)\in\V(\D_2)$. Altogether we have
shown that $\V(\D_1) \subseteq \V(\D_2)$.

Suppose now that $(u,v)\in\V(\D_2)$ as witnessed by a path $\mu$ in $\D_2$.  If
$\mu$ does not go through $x_j$ (i.e. $x_j$ is not an internal node of $\mu$),
 it is also a path in $\D_1$ and
$(u,v)\in\V(\D_1)$.  If $\mu$ goes through $x_j$ but does not end between $x_j$
and $x_k$ we can also conclude that $(u,v)\in\V(\D_1)$ using the fact that
$f_{ik}=f_{jk}$.  It remains to consider the case when $\mu$ ends with $x_j a_j
\ldots a_{\beta-1} x_\beta$ for some $\beta$ with $j < \beta < k$ (in
particular $v=x_\beta$).  As $\Delta_{ij}=\Delta_{jk}$ there exists $\alpha$
with $i < \alpha < j$ such that $\delta_\V(\cdot,\lambda(\pi_{i\rightarrow
  \alpha})) = \delta_\V(\cdot,\lambda(\pi_{j\rightarrow \beta}))$. From this we
can construct a path $\mu'$ in $\D_1$ replacing in $\mu$ the segment $x_j a_j
\ldots a_{\beta-1} x_\beta$ by $x_i a_i \ldots a_{\alpha-1} x_\alpha$,
witnessing the fact that $(u,x_\alpha) \in \V(\D_1)$, a contradiction as
$x_\alpha$ is not an element of $\V(\D_1)$. Altogether we have proved that
$\V(\D_2) \subseteq \V(\D_1)$. Hence, $\V(\D_2) = \V(\D_1)= \V(\D)$.\qedhere
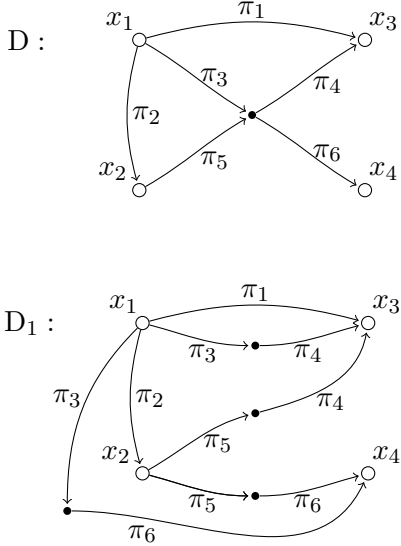
\begin{figure}[h!]
	\centering
	\begin{tikzpicture}[shorten >=1pt,->]
	 	\tikzstyle{vertex}=[circle,fill=black,minimum size=3pt,inner sep=0pt]
	 	\tikzstyle{viewnode}=[draw, circle, minimum size=5pt, inner sep = 0pt]
	 	\foreach \name/\x/\y in {1/0/2, 2/3/2, 3/0/0, 4/3/0} \node[viewnode] (\name) at (\x,\y) {};
		\node[vertex] (5) at (1.5,1) {};
		\draw (1) edge[out=15,in=165] (2) ;
		\draw (1) edge[out=255,in=105] (3) ;
		\draw (1) edge[out=-30,in=150] (5) ;
		\draw (3) edge[out=30,in=210] (5) ;
		\draw (5) edge[out=30,in=210] (2) ;
		\draw (5) edge[out=-30,in=150] (4) ;
		\foreach \name/\x/\y in {1/1.5/2.40, 2/0.1/1, 3/1/1.5, 4/2.5/1.4, 5/1/0.4, 6/2.5/0.5} \node () at (\x,\y) {$\pi_\name$};
		\foreach \name/\x/\y in {1/-0.25/2.25, 2/-0.35/0.25, 3/3.25/2.25, 4/3.25/0.25} \node () at (\x,\y) {$x_\name$};
		\node () at (-1.5,2) {$\D :$};
	\end{tikzpicture}

	\vspace{1cm}

	\begin{tikzpicture}[shorten >=1pt,->]
	 	\tikzstyle{vertex}=[circle,fill=black,minimum size=3pt,inner sep=0pt]
	 	\tikzstyle{viewnode}=[draw, circle, minimum size=5pt, inner sep = 0pt]
	 	\foreach \name/\x/\y in {1/0/2, 2/3/2, 3/0/0, 4/3/0} \node[viewnode] (\name) at (\x,\y) {};
		\node[vertex] (5) at (1.5,0.8) {};
		\node[vertex] (5b) at (1.5,-0.3) {};
		\node[vertex] (5c) at (1.5,1.7) {};
		\node[vertex] (5d) at (-1,-0.5) {};
		\draw (1) edge[out=15,in=165] (2) ;
		\draw (1) edge[out=255,in=105] (3) ;
		\draw (3) edge[out=30,in=195] (5) ;
		\draw (5) edge[out=15,in=260] (2) ;
		\draw (3) edge[out=-15,in=180] (5b) ;
		\draw (3) edge[out=-15,in=180] (5b) ;
		\draw (5b) edge[out=0,in=195] (4) ;
		\draw (1) edge[out=-15,in=180] (5c) ;
		\draw (5c) edge[out=0,in=195] (2) ;
		\draw (1) edge[out=225,in=90] (5d) ;
		\draw (5d) edge[out=0,in=-115] (4) ;
		\foreach \name/\x/\y in {1/1.5/2.40, 2/0.1/1, 3/0.8/1.6, 3/-1/1, 4/2.5/0.95, 4/2.2/1.6, 5/1/0.4, 5/0.8/-0.4, 6/2.2/-0.4, 6/0/-0.8} 
			\node () at (\x,\y) {$\pi_\name$};
		\foreach \name/\x/\y in {1/-0.25/2.25, 2/-0.35/0.25, 3/3.25/2.25, 4/3.25/0.25} \node () at (\x,\y) {$x_\name$};
		\node () at (-1.5,2) {$\D_1 :$};
	\end{tikzpicture}

	\caption{Illustration of the transformation from $\D$ to $\D_1$ in Lemma~\ref{lemma-inverse-view}.
	Nodes are colored white or black depending on whether they appear in $\V(\D)$ or not.}
	\label{figure-inverse-view}
\end{figure}
\end{proof}

In view of Lemma~\ref{lemma-inverse-view}, we know that if \V determines \Q
then there exists a rewriting \R with \np data complexity. Indeed \R is the
query computed by the following non-deterministic polynomial time algorithm: on
an input $\tau$-structure \S, guess from \S a database \D of polynomial size,
check that $\V(\D)=\S$ and then evaluate \Q on \D. There also exists a rewriting 
with \conp data complexity, by considering all databases \D of 
polynomial size such that $\V(\D)=\S$. Altogether we get:

\begin{cor}\label{npconp-result}
  Let \V and \Q be \RPQs such that \V determines \Q. Then there exists a
  rewriting of \Q using \V with \np data complexity, and another one with  
  \conp data complexity.
\end{cor}

It is not known whether, for \RPQ views and queries, determinacy implies the existence of
a rewriting with polynomial time data complexity. The complexity bounds of
Corollary~\ref{npconp-result} are the current best known bounds. We will see in
the next sections that if we further assume that the function induced by \Q
using \V is \emph{monotone} then there exists a rewriting of \Q using \V
definable in \Datalog and therefore computable in polynomial time.

Using a more intricate pumping argument it is possible to show that for any
Conjunctive Regular Path Query view \V, the fact that a view instance is
in the image of \V can also be witnessed by a database of polynomial size. Hence
Corollary~\ref{npconp-result} extends to Conjunctive Regular Path Queries.

However we will see that for Context-Free Path Query views there is no
recursive bound on the size of a database yielding a given view instance. This
will follow from Lemma~\ref{lemma-cf-view-test} showing that,
for Context-Free Path Query views, checking whether a view
instance is in the image of the view is undecidable.

\subsection{Testing for view images}

We now consider the following problem. We are given a view $\V$ from $\sigma$ to $\tau$
and a $\tau$-structure $\S$ and we are asking whether there exists a
$\sigma$-structure \D such that $\V(\D)=\S$. 

Note that this problem is related to the previous one.  In view of
Lemma~\ref{lemma-inverse-view} we immediately get an \np algorithm for testing
membership in the image of an \RPQ view \V: on input $\S$ guess a database $\D$
of size polynomial in $\S$ and check $\V(\D)=\S$.  We will see that testing for
view images is \np-hard for Regular Path Query views and undecidable for
Context-Free Path Query views.

Moreover one can show that if testing for view images can be done in \ptime
then, for \Q and \V such that \V determines \Q, then there exists a rewriting
of \Q using \V with polynomial time data complexity. The polynomial time
algorithm works as follows. On a view instance \S, it first tests whether there
exists a database \D such that $\S=\V(\D)$. If not it rejects. If yes, consider
the schema adding two new letters $a$ and $b$ and consider the query $Q_{a,b}$
asking for a path in the language $a\cdot L(Q) \cdot b$. Define $\V'$ as $\V
\cup \set{\Q_{a,b},V_a,V_b}$ where $V_a$ and $V_b$ return all pairs of nodes
linked by $a$ and $b$ respectively. For each pair $(x,y)$ of nodes of $\S$, let
$\S'$ be \S expanded with the empty relation for $\Q_{a,b}$, a single pair
$(u,x)$ for $V_a$ and a single pair $(y,v)$ for $V_b$ where $u$ and $v$ are two
new nodes. We then test whether $\S'$ is a view image. A simple argument shows
that the test says yes iff $(x,y) \not\in\Q(\D)$ and the algorithm works in
time polynomial in the size of \S.

Unfortunately, as already mentioned, the test for view images is \np-hard
already for \RPQ views.

\begin{lem}\label{lemma-test}
  There is an \RPQ view \V from $\sigma$ to $\tau$ such that given a
  $\tau$-structure \S it is \np-hard to test whether there exists a
  $\sigma$-structure \D such that $\V(\D)=\S$.
\end{lem}
\begin{proof}
  We reduce 3-\textsc{Colorability} to our problem.  The proof is a simple
  variation of the reduction found in \cite{calvanese2000answering} to prove
  that computing certain answers under the sound view assumption is \conp-hard
  in data complexity. 

        Let $\sigma = \{rg,gr,bg,gb,rb,br\}$ and $\tau=\{V_1, V_2\}$. By abuse of notation, we will
        refer to an element of $\sigma$ as $\alpha\beta$, with $\alpha$ and
        $\beta$ two symbols in $\{r,g,b\}$, and $\alpha \neq \beta$.  Let \V be
        the following view from $\sigma$ to $\tau$:

\begin{itemize}
	\item $\V = \set{V_1,V_2}$
        \item $L(V_1) =
                \set{rg,gr,bg,gb,rb,br}$
        \item $L(V_2) = \set{\alpha_1\beta_1\cdot\alpha_2\beta_2 \ | \ \beta_1 \neq
                \alpha_2}$.
            \end{itemize}
            
            Let $G = (U,W)$ be a connected graph. From $G$ we define a
            $\tau$-structure $\S_G$, in which the interpretation of $V_1$ is:
$$\set{(x,y) \ | \ (x,y) \in W \textrm{ or } (y,x) \in W}$$
and the interpretation of $V_2$ is the empty relation.

We show that $G$ is 3-colorable iff there exists $\D$ such that $\V(\D)=E_G$.
Intuitively, the idea is that $\sigma$ describes the colors of the edges of G,
that is the color of the two end points of each edge. For instance, if $x$ and
$y$ are linked by $rg$, then it should be understood that $x$ is red and $y$ is
green. $V_1$ checks that each pair of nodes that are connected in $G$ are
colored with (at least) two different colors, and $V_2$ checks if there is any
error, that is, if a node is required to have more than one color. Since $V_2$
is empty, any graph database \D such that $\V(\D) = \S$ cannot have any
such error, and would thus be 3-colorable.

More precisely, assume that $G$ is 3-colorable. Then there exists a coloring
function $c : U \rightarrow \{r,g,b\}$ such that $c(x)
\neq c(y)$ for all $(x,y) \in W$. We define \D as the $\sigma$-structure such that, 
for each $\alpha\beta \in \sigma$, the interpretation of $\alpha\beta$ in $\D$
is:
\begin{align*}
\set{(x,y) \ | \ & (x,y) \in W \textrm{ or } (y,x) \in W,\\
&\textrm{ and } c(x) = \alpha, c(y) = \beta }.
\end{align*}
It is then easy to check that $\V(\D) = \S_G$. Indeed, for all $x,y,z \in \D$,
if $\alpha_1\beta_1(x,y)$ and $\alpha_2\beta_2(y,z)$ hold in $\D$, then
$\beta_1 = c(y) = \alpha_2$, hence $(x,z) \notin V_2(\D)$, so $V_2(\D)$ is
empty.

Conversely, assume that there exists a graph database \D such that $\V(\D) =
\S_G$. Consider the coloring function $c : U \rightarrow \{r,g,b\}$ defined as:
$c(x) = \alpha$ if there exists $y$ such that $\alpha\beta(x,y)$ holds in \D. 
Since $V_2(\D)$ is empty, it is immediate to check that $c(x)$ is  uniquely 
defined and that $c$ is a proper $3$-coloring of $G$.
\end{proof}

If we go from regular languages to context-free ones, then the problem becomes
undecidable.

\begin{lem}
  Let $\V$ be a Context-Free Path Query view from $\sigma$ to $\tau$. Let $\S$ be a
  $\tau$-instance. Then it is undecidable whether there exists a
  $\sigma$-structure $\D$ such that $\V(\D) = \S$.
\end{lem}

\begin{proof}
  We prove this by reduction from the universality problem for context-free
  languages. Let $L$ be a context-free language over some alphabet
  $\sigma$. Let $\$$ be a fresh symbol that does not appear in $\sigma$. Let
  $\V = \set{V_1, V_2}$, where $V_1$ is defined by $L(V_1) = \$\cdot L \cdot
  \$$ and $V_2$ is defined by $L(V_2) = \$\cdot \sigma^* \cdot \$$. Finally,
  let $\S$ be the view instance that contains a single pair $(x,y)$ in $V_2$
  and no pair in $V_1$. Then there exists $\D$ such that $\V(\D) = \S$ if and
  only if $L$ is not universal over $\sigma$.

  \begin{itemize}
  \item Assume that there exists a database $\D$ such that $\V(\D) = \S$. Then
    there exists a path $\pi$ from $x$ to $y$ such that $\lambda(\pi) \in
    L(V_2)$. Hence there exists $w\in \sigma^*$ such that $\lambda(\pi) =
    \$\cdot w\cdot\$$. However, $\lambda(\pi) \notin L(V_1)$. Hence $w \notin
    L$ and $L$ is not universal.

  \item Conversely, assume that $L$ is not universal. Then there exists $w \in
    \sigma^*$ such that $w\notin L$. Then it is easy to check that the database
    $\D$ consisting of a simple path labeled by $\$\cdot w \cdot\$$ satisfies
    $\V(\D) = \S$.\qedhere
	\end{itemize}
\end{proof}	

\noindent A more intricate argument shows that undecidability already holds for a \emph{fixed} view definition~\V.

\begin{lem}\label{lemma-cf-view-test}
  There exists a fixed Context-Free Path Query view \V from $\sigma$ to $\tau$ such that, given a
  $\tau$-structure \S, it is undecidable whether there exists a $\sigma$-structure \D
  such that $\V(\D) = \S$.
\end{lem}

\begin{proof}
  Let $\sigma = \set{(,;,),a,b,\$,1}$. Let $\boldsymbol{\sigma}$ be a copy of
  $\sigma$ with fresh symbols. For $\alpha \in \sigma$, we denote by
  $\boldsymbol{\alpha}$ the corresponding symbol in $\boldsymbol{\sigma}$.  For
  $w$ a word, $\tilde{w}$ denote the word corresponding to $w$ read from right
  to left.  $\V$ consists of views that reveal each symbol in $\sigma$, that
  is, for all $\alpha \in \sigma$, $\V$ contains a view $V_\alpha$ defined by
  $L(V_\alpha) = \set{\alpha}$. Additionally, $\V$ contains the queries
	$V_u$, $V_v$, $V'_u$, $V'_v$, $V_g$ and $V_c$ defined by the following 
	equations:
  	$$
		L(V_u) = \left \{
		\begin{array}{cc}
		 \boldsymbol{\$} \cdot w \cdot \boldsymbol{\$} \cdot \po i_1 \pv v_1 \pv u_1
      \pf \ldots \po i_n \pv v_n \pv u_n \pf \cdot \boldsymbol{\$} \ |\\ 
			 w,u_k,v_k \in
      \set{\boldsymbol a,\boldsymbol b}^*, i_k \in \boldsymbol 1^*,
      u_1\cdot\ldots\cdot u_n = \tilde w
		\end{array}
		\right \}
	$$
	$$
		L(V_v) = \left \{
		\begin{array}{cc}
		 	\boldsymbol{\$} \cdot w \cdot \boldsymbol{\$} \cdot \po i_1 \pv v_1 \pv u_1
      \pf \ldots \po i_n \pv v_n \pv u_n \pf \cdot \boldsymbol{\$} \ |\\ 
			 w,u_k,v_k \in
      \set{\boldsymbol a,\boldsymbol b}^*, i_k \in \boldsymbol 1^*,
      v_1\cdot\ldots\cdot v_n = \tilde w
		\end{array}
		\right \}
	$$
	$$
		L(V'_u) = \left \{
		\begin{array}{cc}
			\boldsymbol{\$} \cdot w \cdot \boldsymbol{\$} \cdot \po i_1 \pv v_1 \pv
      u_1 \pf \ldots \po i_n \pv v_n \pv u_n \pf \cdot \boldsymbol{\$} \ |\\ 
			w,u_k,v_k \in
      \set{\boldsymbol a,\boldsymbol b}^*, i_k \in \boldsymbol 1^*,
      u_1\cdot\ldots\cdot u_n \neq \tilde w
		\end{array}
		\right \}
	$$
	$$
		L(V'_u) = \left \{
		\begin{array}{cc}
			\boldsymbol{\$} \cdot w \cdot \boldsymbol{\$} \cdot \po i_1 \pv v_1 \pv
      u_1 \pf \ldots \po i_n \pv v_n \pv u_n \pf \cdot \boldsymbol{\$} \ |\\ 
			w,u_k,v_k \in
      \set{\boldsymbol a,\boldsymbol b}^*, i_k \in \boldsymbol 1^*,
      v_1\cdot\ldots\cdot v_n \neq \tilde w
		\end{array}
		\right \}
	$$
	$$
		L(V_g) = \left \{
		\begin{array}{cc}
			\$ \cdot (u_1;v_1;i_1) \cdot \ldots \cdot
      (u_n;v_n;i_n) \cdot \boldsymbol{\$} \cdot \boldsymbol{\sigma}^* \cdot \boldsymbol{\$} \cdot
      \boldsymbol{\sigma}^* \cdot \po i' \pv v' \pv u' \pf \ |\\ 
			u_k,v_k \in
      \set{a,b}^*, i_k\in 1^*, u',v' \in \set{\boldsymbol a, \boldsymbol b}^*,
      i' \in \boldsymbol 1^*, i' > i_n
		\end{array}
		\right \}
	$$
	$$
		L(V_c) = \left \{
		\begin{array}{ccc}
			\$ \cdot (u_1;v_1;i_1) \cdot \ldots \cdot
      (u_n;v_n;i_n) \cdot \boldsymbol{\$} \cdot \boldsymbol{\sigma}^* \cdot \boldsymbol{\$} \cdot
      \boldsymbol{\sigma}^* \cdot \po i' \pv v' \pv u' \pf \ |\\ 
			u_k,v_k \in
      \set{a,b}^*, i_k\in 1^*, u',v' \in \set{\boldsymbol a, \boldsymbol b}^*,
      i' \in \boldsymbol 1^*,\\ 
			\exists k, i_k = \varphi(i'), u_k \neq
      \varphi(\tilde u') \text{ or } v_k \neq \varphi(\tilde v')
		\end{array}
		\right \}
	$$
	where $\varphi$ is the function that maps each symbol in $\boldsymbol \sigma$ 
	to the corresponding symbol in $\sigma$.

  One can check that all these languages are actually context-free languages.

  We now prove that, given a view instance $\S$ for this specific view $\V$, it
  is undecidable whether there exists a database $\D$ such that $\V(\D) =
  \S$. We prove this by reduction from the Post Correspondence Problem
  (PCP). Let $(u_i,v_i,i)_{0 < i \leq n}$ be an instance of PCP over
  $\set{a,b}$, where the third argument explicitly gives the index of each
  pair. We build the following view instance $\S$:

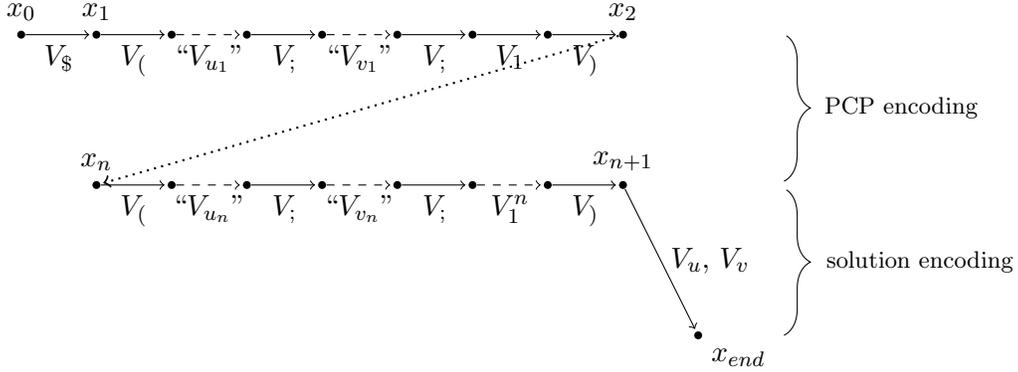
	\begin{center}
	\begin{tikzpicture}[shorten >=1pt, ->]
		\tikzstyle{vertex}=[circle,fill=black,minimum size=3pt,inner sep=0pt]
			\node[vertex,label=above:{$x_0$}] (1) at (0,0) {};
			\node[vertex,label=above:{$x_1$}] (2) at (1,0) {};
			\draw (1) edge node[below] {$V_\$$} (2);
			\node[vertex] (3) at (2,0) {};
			\draw (2) edge node[below] {$V_($} (3);
			\node[vertex] (4) at (3,0) {};
			\draw (3) edge[dashed] node[below] {$``V_{u_1}"$} (4);
			\node[vertex] (5) at (4,0) {};
			\draw (4) edge node[below] {$V_;$} (5);
			\node[vertex] (6) at (5,0) {};
			\draw (5) edge[dashed] node[below] {$``V_{v_1}"$} (6);
			\node[vertex] (7) at (6,0) {};
			\draw (6) edge node[below] {$V_;$} (7);
			\node[vertex] (8) at (7,0) {};
			\draw (7) edge node[below] {$V_1$} (8);
			\node[vertex, label=above:{$x_2$}] (9) at (8,0) {};
			\draw (8) edge node[below] {$V_)$} (9);
			\node[vertex, label=above:{$x_n$}] (11) at (1,-2) {};
			\draw (9) edge[dotted, thick] (11);
			\node[vertex] (12) at (2,-2) {};
			\draw (11) edge node[below] {$V_($} (12);
			\node[vertex] (13) at (3,-2) {};
			\draw (12) edge[dashed] node[below] {$``V_{u_n}"$} (13);
			\node[vertex] (14) at (4,-2) {};
			\draw (13) edge node[below] {$V_;$} (14);
			\node[vertex] (15) at (5,-2) {};
			\draw (14) edge[dashed] node[below] {$``V_{v_n}"$} (15);
			\node[vertex] (16) at (6,-2) {};
			\draw (15) edge node[below] {$V_;$} (16);
			\node[vertex] (17) at (7,-2) {};
			\draw (16) edge[dashed] node[below] {$V_1^n$} (17);
			\node[vertex, label=above:{$x_{n+1}$}] (18) at (8,-2) {};
			\draw (17) edge node[below] {$V_)$} (18);
			\node[vertex, label=below right:{$x_{end}$}] (22) at (9,-4) {};
			\draw (18) edge node[right] {$V_u$, $V_v$} (22) ;

			\draw [decorate,decoration={brace,amplitude=10pt,mirror,raise=4pt},-]
				(10,-1.95) -- (10,0) node [black,midway,xshift=1.7cm] {\footnotesize{PCP encoding}};
			\draw [decorate,decoration={brace,amplitude=10pt,mirror,raise=4pt},-]
				(10,-4) -- (10,-2.05) node [black,midway,xshift=1.95 cm] {\footnotesize{solution encoding}};

	\end{tikzpicture}
	\end{center}

	We now show that there exists $\D$ such that $\V(\D) = \S$ if and only
        if the PCP instance is satisfiable.  Intuitively, $\S$ consists of two
        parts. The first part, from $x_0$ to $x_{n+1}$ is the encoding of the
        PCP instance. It uses letters from $\sigma$ that are all revealed by
        the view. All tuples are simply enumerated in the natural order, where
        the $i$th tuple is encoded between $x_i$ and $x_{i+1}$. The dashed
        arrows $V_{u_i}$ and $V_{v_i}$ represent the correct succession of
        $V_a$ and $V_b$ that naturally encode $u_i$ and $v_i$, whereas the
        $V_1^i$ part is the unary encoding of $i$, the index of the tuple. The
        second part of the instance states the existence of a solution for this
        instance, and uses ``hidden" letters from $\boldsymbol \sigma$. $V_u$
        and $V_v$ states that there exists a solution, and the fact that all
        other views are empty checks that this solution is correct.
	\begin{itemize}
        \item Assume that there exists a database $\D$ such that $\V(\D) =
          \S$. Then there exists a path $\pi$ from $x_{n+1}$ to $x_{end}$ such
          that $\lambda(\pi) \in L(V_u)$. Hence, this path is of the form $\$
          \cdot w \cdot \$ \cdot \po i_{1} \pv v'_{1} \pv u'_{1} \pf \ldots \po
          i_{m} \pv v'_{m} \pv u'_{m} \pf \cdot \$$, where $w$ is a word in
          $\boldsymbol \sigma^*$ and $u'_{1} \ldots u'_{m} = \tilde w$. Remark
          that is also holds that $v'_{1} \ldots v'_{m} = \tilde w$, otherwise
          $\lambda(\pi) \in V'_v$, which would imply that $(x_{n+1},x_{end})\in
          V'_v(\D)$, and lead to a contradiction.

          Hence, $u'_1\ldots u'_m = v'_1 \ldots v'_m$. It remains to show that
          each $\po i_{i} \pv v'_{i} \pv u'_{i} \pf$ is an encoding of the
          mirror of some tuple in the PCP instance, which would imply a
          solution as $\tilde u'_m\ldots \tilde u'_1 = \tilde v'_m \ldots
          \tilde v'_1$. In other words, $u_{\size{i_m}}\ldots u_{\size{i_1}} =
          v_{\size{i_m}}\ldots v_{\size{i_1}}$.

          Assume that one of the $\po i_{i} \pv v'_{i} \pv u'_{i} \pf$ is not
          the mirror of some tuple encoded in the first half of the
          instance. Remark that $\size{i_i} \leq n$. Otherwise, there exists a
          path whose label is in $L(V_g)$, which leads to a
          contradiction. Hence, either $u'_i \neq \tilde u_{\size{i_i}}$ or
          $v'_i \neq \tilde v_{\size{i_i}}$.  Both cases lead to the existence
          of a path whose label is
          in $L(V_c)$, and thus to a contradiction.

        \item Assume that there exists a solution $i_1\ldots i_m$ to the PCP
          instance. Then the database $\D$ that consists of the following
          simple path is such that $\V(\D) = \S$:
$$\$(u_1;v_1;1)\ldots(u_n;v_n;1^n)\boldsymbol{\$}\bold{u_{i_1}\ldots u_{i_m}}\boldsymbol{\$}
\po \bold 1^{i_m} \pv \bold{\tilde v_{i_m}} \pv \bold{\tilde u_{i_m}} \pf
\ldots \po \bold 1^{i_1} \pv \bold{\tilde v_{i_1}} \pv \bold{\tilde u_{i_1}}
\pf \boldsymbol{\$}$$
 where $\bold{u_i}$ and $\bold{v_i}$ simply represent the corresponding
$u_i$ and $v_i$ written using $\boldsymbol a$ and $\boldsymbol b$ instead of
$a$ and $b$.\qedhere
	\end{itemize}
\end{proof}

\noindent Note that in the proof of Lemma~\ref{lemma-cf-view-test} the view instance is a
coding of a PCP instance and the corresponding database a coding of a
solution. As there is no recursive bound on the size of a solution of a PCP
instance, for Context-Free Path Query views, there are no recursive bound on
the size of a database that yields a given view instance. This is to be
compared with the polynomial bound for \RPQ views shown in
Lemma~\ref{lemma-inverse-view}.

\section{Monotone determinacy and rewriting}\label{section-monotone}

As Example~\ref{example-determinacy} shows, there is an \RPQ view \V and an \RPQ
query \Q such that \V determines \Q but the function induced by \Q using \V is not
monotone, therefore having no \RPQ rewriting. It is natural to wonder
whether the monotonicity of the function induced by the query is the only limit
for the existence of an \RPQ rewriting. Recall from the introduction that if \V
and \Q are defined using \CQ{}s and \V determines \Q, then the function induced
by \Q using \V is monotone iff there exists a \CQ rewriting.  In the case of
\RPQ views and queries the analog does not hold. We will see that, even if we
assume monotonicity, an \RPQ rewriting need not exist; however in the next
section we will show that a rewriting definable in \Datalog always exists. We
start by formalizing the notion of monotone determinacy.

\begin{defi}[Monotone determinacy]
  We say that a view \V determines a query \Q in a monotone way if \V
  determines \Q and the function induced by \Q using \V is monotone.
\end{defi}

It is rather immediate to see that monotone determinacy is equivalent to the
following property for \V and \Q:
\begin{equation*}
\forall \D,\D', \quad \V(\D) \subseteq \V(\D') \ \Rightarrow \ \Q(\D)
\subseteq \Q(\D')
\end{equation*}
This turns out to coincide with the notion of \emph{losslessness under the
  sound view assumption} defined in~\cite{calvanese2002lossless}, that was
shown to be decidable, actually \expspace-complete, for \RPQ{}s.

\begin{cor}\label{cor:mon-det}
	The monotone determinacy problem for \RPQ{}s is \expspace-complete.
\end{cor} 

Note that in the proof of Proposition~\ref{prop-cfpq-determinacy}, the
rewriting is always monotone when the view determines the query. Therefore, for
Context-Free Path Query views and \RPQ queries, monotone determinacy is
undecidable.

Recall from Example~\ref{example-determinacy} that there exist a view and a query
such that the view determines the query but not in a monotone way.
We now assume given an \RPQ view \V and an \RPQ query \Q such that \V
determines \Q in a monotone way.  It was observed in~\cite{calvanese2002lossless}
that even in this case there might be no rewriting definable in \RPQ.

In fact, given \V and \Q defined using \RPQs, it is decidable whether an \RPQ
rewriting exists and the problem is
\twoexpspace-complete~\cite{calvanese02jcss}. As testing monotone determinacy
is \expspace-complete, a simple complexity argument shows that an \RPQ
rewriting is not guaranteed to exist under monotone determinacy.

Here is a concrete example witnessing this fact.\footnote{A similar example was 
claimed in~\cite[Example 4]{calvanese2002lossless} but it seems that in this 
example \V and \Q are such that \V does not determine \Q.}

\begin{myexample}\label{example-norpq}
Let $\sigma=\set{a,b,c}$. Let \Q and \V be defined as follows:
	\begin{itemize}
		\item $\Q = ab^*a \ | \ ac^*a$
		\item $\V=\set{V_1,V_2,V_3}$ with
			\begin{itemize}
				\item $V_1 = ab^*$
				\item $V_2 = ac^*$ 
				\item $V_3 = b^*a \ | \ c^*a$
			\end{itemize}
	\end{itemize}
One can verify that \V determines \Q as witnessed by the following
rewriting $\R(x,y)$:
\begin{equation*}
\exists z ~ V_1(x,z) \wedge V_2(x,z) \wedge V_3(z,y)
\end{equation*}
That \R is a rewriting is illustrated in Figure~\ref{figure-norpq}.  Consider
the database \D of Figure~\ref{figure-norpq} which is a typical database such
that $(x,y) \in \Q(\D)$. The choice of $z$ witnessing $(x,y) \in
R(\V(\D))$ is then immediate. Conversely, consider the database $\D'$ of
Figure~\ref{figure-norpq}. It is a typical database such that $(x,y) \in
R(\V(\D))$. The top path shows that $(x,y)~\in~\Q(\D)$.

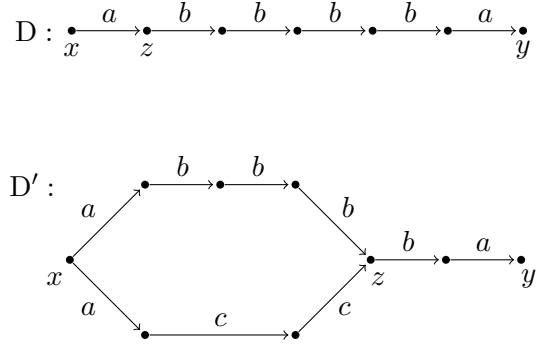
\begin{figure}[h!]
	\centering
	\begin{tikzpicture}[shorten >=1pt,->]
	 	\tikzstyle{vertex}=[circle,fill=black,minimum size=3pt,inner sep=0pt]
	 	\foreach \name/\x/\y in {1/0/0, 2/1/0, 3/2/0, 4/3/0, 5/4/0, 6/5/0, 7/6/0} \node[vertex] (\name) at (\x,\y) {};
		\foreach \from/\to in {1/2, 2/3, 3/4, 4/5, 5/6, 6/7} \draw (\from) -- (\to);
	  \node () at (0.5,0.2) {$a$} ; \node () at (5.5,0.2) {$a$} ;
		\foreach \x in {1.5, 2.5, 3.5, 4.5} \node () at (\x,0.25) {$b$} ;
	 	\foreach \name/\x in {x/0, z/1, y/6} \node (\name) at (\x,-0.25) {$\name$};
		\node (D) at (-0.5,0) {$\D :$};
	\end{tikzpicture}

	\vspace{1cm}

	\begin{tikzpicture}[shorten >=1pt,->]
		\tikzstyle{vertex}=[circle,fill=black,minimum size=3pt,inner sep=0pt]
		\foreach \name/\x/\y in {1/0/0, 2a/1/1, 2b/1/-1, 3a/2/1, 3b/3/-1, 4a/3/1, 5/4/0, 6/5/0, 7/6/0} 
			\node[vertex] (\name) at (\x,\y) {};
		\foreach \from/\to in {1/2b, 2b/3b, 3b/5, 5/6, 6/7} \draw (\from) -- (\to);
		\foreach \from/\to in {1/2a, 2a/3a, 3a/4a, 4a/5} \draw (\from) -- (\to);
		\node () at (0.25,0.65) {$a$} ; \node () at (0.25,-0.65) {$a$} ;
		\foreach \x in {1.5, 2.5} \node () at (\x,1.25) {$b$} ;
 		\node () at (2,-0.8) {$c$} ;
		\node () at (3.7,0.7) {$b$} ;
		\node () at (3.65,-0.65) {$c$} ;
		\node () at (4.5,0.25) {$b$} ; \node () at (5.5,0.20) {$a$} ;
	 	\foreach \name/\x in {x/-0.2, z/4.1, y/6.1} \node (\name) at (\x,-0.25) {$\name$};
		\node () at (-0.5,1) {$\D' :$} ;
	\end{tikzpicture}
	\caption{Databases $\D$ and $\D'$ for Example~\ref{example-norpq}.}
	\label{figure-norpq}
\end{figure}

\noindent Since \R is monotone, \V determines \Q in a monotone way.  It can also be shown
(for instance using the decision procedure provided in~\cite{calvanese02jcss})
that no \RPQ rewriting exists.
\end{myexample}

In the previous example we have exhibited a Conjunctive Regular Path Query
rewriting. However the following example suggests that Conjunctive Regular Path Query is not expressive
enough as a rewriting language.

\begin{myexample}\label{example-nocrpq}
  Let $\sigma=\set{a}$. Let \V and \Q be defined as follows:
	\begin{itemize}
        \item $\Q = a(a^6)^* \ | \ aa(a^6)^*$
\hfill (words of length 1 or 2 modulo 6)

	\item $\V=\set{V_1,V_2}$ with
			\begin{itemize}
				\item $V_1 = a \ | \ aa$\hfill (words of length 1 or 2)
				\item $V_2 = aa \ | \ aaa$\hfill (words of length 2 or 3) 
			\end{itemize}
	\end{itemize}

       \noindent It can be verified that \V determines \Q in a monotone way as witnessed by the following
         rewriting $\R(x,y)$:
\begin{equation*}
 \exists z ~ V_1(x,z) \wedge T^*(z,y)
\end{equation*}
where $T(x,y)$ is defined as:
\begin{align*}
 \exists z_1,z_2 ~ & V_1(x,z_1) \wedge V_2(x,z_1) \wedge V_1(z_1,z_2) \wedge\\
                  &  V_2(z_1,z_2) \wedge V_1(z_2,y) \wedge V_2(z_2,y) 
\end{align*}
The query $T$ is such that if $T(x,y)$ holds in $\V(\D)$, then in $\D$ the nodes $x$ and $y$ are either
linked by a path of length~$6$ or by both a path of length~$5$ and a path of
length~$7$. This fact can be checked by a simple case analysis. One such case is
illustrated in Figure~\ref{figure-nocrpq}. In this case there is no path of
length~$6$ in $\D$, but the top path has length~$5$, and the path starting with the
bottom segment and then the last two top segments has length~$7$.

From this, a simple induction shows that if $T^*(x,y)$ holds in $\V(\D)$, then in $\D$ 
the nodes $x$ and $y$
are either linked by a path of length $0$ modulo $6$, or by both a path of length
$1$ modulo $6$ and a path of length $5$ modulo $6$. 

Assume now that $\R(x,y)$ holds in $\V(\D)$.  Then in $\D$ there exists a $z$ such that $x$ is at
distance $1$ or $2$ from $z$, and such that $T^*(z,y)$ holds in $\V(\D)$. Assume first that
$z$ and $y$ are at distance $0$ modulo $6$ in $\D$. In this case, regardless of the
distance between $x$ and $z$, $\Q(x,y)$ holds in $\D$. Otherwise, in $\D$ there exist both a
path of length $1$ modulo $6$ and a path of length $5$ modulo $6$ from $z$ to
$y$. Therefore, if $x$ and $z$ are at distance $1$, the first path from $z$ to
$y$ yields a path of length $2$ modulo $6$ and, if $x$ and $z$ are at
distance $2$, the second path from $z$ to $y$ yields a path of length $1$
modulo $6$, see Figure~\ref{figure-nocrpq2}.

Conversely, it is easy to check that $\R(x,y)$ holds in $\V(\D)$ whenever $\Q(x,y)$ holds in $\D$. This 
follows from the fact that $T(x,y)$ holds in $\V(\D)$ for all $x$ and $y$ that are at distance $6$ in $\D$.

Notice that \R is monotone. A tedious combinatorial argument can show that \R
cannot be expressed as a Conjunctive Regular Path Query.
\end{myexample}

\begin{figure}[h!]

	\centering
	\begin{tikzpicture}[shorten >=1pt,->]
	 	\tikzstyle{vertex}=[circle,fill=black,minimum size=3pt,inner sep=0pt]
	 	\foreach \name/\x in {1/0, 2/2, 3/4, 4/6} \node[vertex] (\name) at (\x,0) {};
		\foreach \from/\to in {1/2, 2/3, 3/4} \draw (\from) edge[out=60,in=120] (\to);
		\foreach \from/\to in {1/2, 2/3, 3/4} \draw (\from) edge[out=-60,in=-120] (\to);
	 	\foreach \name/\x in {x/-0.25, y/6.25} \node (\name) at (\x,-0.25) {$\name$};
	 	\foreach \name/\x in {1/2, 2/4} \node (\name) at (\x,-0.5) {$z_\name$};
		\node () at (1,0.8) {$V_1:a$};
		\node () at (3,0.8) {$V_1:aa$};
		\node () at (5,0.8) {$V_1:aa$};
		\node () at (1,-0.8) {$V_2:aaa$};
		\node () at (3,-0.8) {$V_2:aa$};
		\node () at (5,-0.8) {$V_2:aa$};
	\end{tikzpicture}

	\caption{Example~\ref{example-nocrpq}: An arbitrary database $\D$ whose view satisfies $T(x,y)$. Each arrow of the
          form $V_i: w$ from a node $u$ to a node $v$ should be understood as a path 
          from $u$ to $v$  whose label is $w$ which witnesses  $(u, v) \in V_i(D)$.}
	\label{figure-nocrpq}
\end{figure}
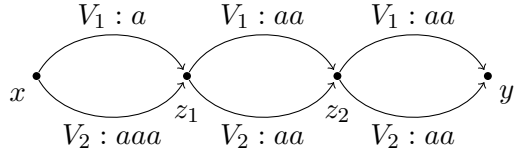

\begin{figure}[h!]

	\centering

	\begin{tikzpicture}[shorten >=1pt,->]
	 	\tikzstyle{vertex}=[circle,fill=black,minimum size=3pt,inner sep=0pt]
	 	\foreach \name/\x in {1/0, 2/2, 3/6} \node[vertex] (\name) at (\x,0) {};
		\draw (1) -- (2);
		\draw (2) -- (3);
	 	\foreach \name/\x in {x/-0.25, y/6.25} \node (\name) at (\x,-0.25) {$\name$};
	 	\node () at (2,-0.25) {$z$};
		\node () at (1,0.25) {$V_1:a \textrm{ or } a^2$};
		\node () at (4,0.25) {$T^*:(a^6)^*$};
	\end{tikzpicture}

	\vspace{1cm}

	\begin{tikzpicture}[shorten >=1pt,->]
	 	\tikzstyle{vertex}=[circle,fill=black,minimum size=3pt,inner sep=0pt]
	 	\foreach \name/\x in {1/0, 2/2, 3/6} \node[vertex] (\name) at (\x,0) {};
		\draw (1) -- (2);
		\draw (2) edge[out=60,in=120] (3);
		\draw (2) edge[out=-60,in=-120] (3);
	 	\foreach \name/\x in {x/-0.25, y/6.25} \node (\name) at (\x,-0.25) {$\name$};
	 	\node () at (2,-0.25) {$z$};
		\node () at (1,0.25) {$V_1:a$};
		\node () at (4,1.25) {$T^*:a(a^6)^*$};
		\node[color = gray] () at (4,-1.30) {$T^*:a^5(a^6)^*$};
	\end{tikzpicture}

	\vspace{1cm}

	\begin{tikzpicture}[shorten >=1pt,->]
	 	\tikzstyle{vertex}=[circle,fill=black,minimum size=3pt,inner sep=0pt]
	 	\foreach \name/\x in {1/0, 2/2, 3/6} \node[vertex] (\name) at (\x,0) {};
		\draw (1) -- (2);
		\draw (2) edge[out=60,in=120] (3);
		\draw (2) edge[out=-60,in=-120] (3);
	 	\foreach \name/\x in {x/-0.25, y/6.25} \node (\name) at (\x,-0.25) {$\name$};
	 	\node () at (2,-0.25) {$z$};
		\node () at (1,0.25) {$V_1:aa$};
		\node[color = gray] () at (4,1.25) {$T^*:a(a^6)^*$};
		\node () at (4,-1.30) {$T^*:a^5(a^6)^*$};
	\end{tikzpicture}

	\caption{The three cases of Example~\ref{example-nocrpq}. The parts that
          are not used for \Q are shaded out.}
	\label{figure-nocrpq2}
\end{figure}
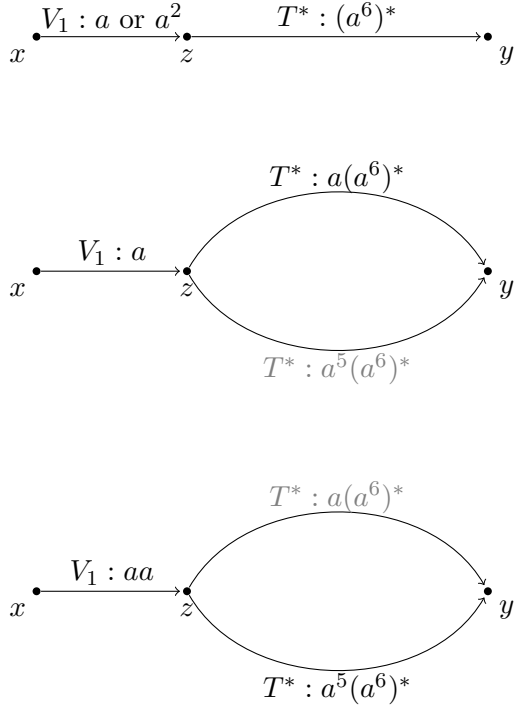

\begin{rem}\label{remark-mso}
  The careful reader has probably noticed that in both  examples above a
  rewriting can be expressed in \MSO. As we will see later, it easily follows
  from the results of~\cite{calvanese2000view} that this is always true in
  general: if \V and \Q are defined by \RPQ{}s and \V determines \Q in a
  monotone way, then there exists a rewriting of \Q using \V definable in \MSO\
   (actually universal \MSO).
\end{rem}

\section{Datalog rewriting}\label{section-datalog}
In this section we prove our main result, namely:

\begin{thm}\label{thm:datalog-rewr}
  If $\V$ and $\Q$ are \RPQs and $\V$ determines $\Q$ in a monotone way then
there exists a \Datalog rewriting of $\Q$ using $\V$.
\end{thm}

Theorem~\ref{thm:datalog-rewr} also implies that the monotone determinacy
problem for \RPQs coincides with the problem of the existence of a \Datalog
rewriting. The latter is therefore decidable by Corollary~\ref{cor:mon-det}:

\begin{cor}
  Let $\V$ and $\Q$ be \RPQs. It is decidable,
  \expspace-complete, whether there exists a \Datalog rewriting of $\Q$ using
  $\V$.
\end{cor}

Our proof being constructive, the \Datalog rewriting can be computed from \V
and \Q.

\paragraph{Main idea and sketch of the proof}
The starting point is the relationship between rewriting and certain answers
under monotone determinacy. One can easily show that if the view determines the
query in a monotone way then the certain answers query is a rewriting. However
certain answers for \RPQ views and queries are \conp-hard to
compute~\cite{calvanese2000answering}. Here we show that there exists
another rewriting (which of course coincides with certain answers on view
images) that is expressible in \Datalog. This other rewriting is suggested by
the relationship between certain answers and Constraint Satisfaction
  Problems (\CSP). Following~\cite{calvanese2000view} we adopt here the
homomorphism point of view for \CSP{}s: Each \CSP is defined by a
structure, called \emph{the template}, and its solutions are all the structures
mapping homomorphically into the template.

Indeed~\cite{calvanese2000view} showed that, for \RPQs $\V$ and $\Q$, certain
answers can be expressed as a \CSP whose template depends only on $\V$ and
$\Q$. It is known from~\cite{feder1998computational} that for every $l$ and $k$
with $l \leq k$, and every template, there exists a \Dataloglk query
approximating the \CSP defined by this template.  Even if its \Dataloglk
``approximation'' does not compute precisely the \CSP associated to \V and \Q,
if it is exact \emph{on view images}, then it is a rewriting.
We show that if the view determines the
query in a monotone way then there is an $l$ and a $k$, depending only on \V
and \Q, such that the \Dataloglk approximation is exact \emph{on
  view images}. This proves the existence of a \Datalog rewriting.

This is done in two steps. We first show that there exists a \Datalog
approximation which is exact on view images of simple path databases.  Then
we show how to lift this result on all view images.  The first step is proved by
a careful analysis of the properties of view images of simple path
databases. The second steps exploits monotonicity.

We now provide more details.

\subsection{Monotone rewritings, certain answers and \CSP}\label{section-certain}

Let \V be a view from $\sigma$ to $\tau$ and \Q be a query on
$\sigma$-structures.  The certain answers of $\Q$ on a $\tau$-structure $\S$
w.r.t. \V are defined as
\begin{equation*}
\cert_{\Q, \V}(\S)=\bigcap_{\D \ | \
  \S \subseteq \V(\D)} \Q(\D)
\end{equation*}
This notion is usually referred to as certain answers under the \emph{sound
  view assumption} or \emph{open world assumption} in the
literature~\cite{AbiteboulD98, calvanese07tcs}.  It is straightforward to check
that if \V determines \Q in a monotone way, the query $\cert_{\Q, \V}$ is a
rewriting of \Q using \V, i.e. $\cert_{\Q, \V}(\V(\D))=\Q(\D)$ for each
$\sigma$-structure \D.

Therefore any language known to express certain answers is a suitable rewriting language under monotone determinacy. 

The following proposition, proved in~\cite{calvanese2000view}, shows that, for
\RPQ views and queries, certain answers (and therefore rewritings) can be
expressed as (the negation of) a \CSP.

\begin{prop}[\cite{calvanese2000view}]\label{prop:cert-csp}
  Let \V be an \RPQ view from $\sigma$ to $\tau$ and \Q be an \RPQ query over
  $\sigma$. There exists a $\tau$-structure $\tempQV$ having a
  set of distinguished \emph{source nodes} and a set of distinguished
  \emph{target nodes} such that, if $\V$ determines $\Q$ in a monotone way, the
  following are equivalent, for each $\sigma$-structure \D and each
  pair of nodes $u, v$ of $\D$:

\begin{enumerate}
\item 
$(u,v) \in \Q(\D)$
\item\label{cert}
$(u, v) \in \cert_{\Q,\V}(\V(\D))$
\item\label{csp} 
$\V(\D)$ has no homomorphism to $\tempQV$ sending $u$ to a source node and $v$ to a target node.\footnote{More precisely~\cite{calvanese2000view} further proved that \ref{cert}. and \ref{csp}. are 
equivalent not only for $\V(\D)$ but for all $\tau$-structures, and even without the assumption that \V determines \Q in a monotone way.}
\end{enumerate}
\end{prop}

In the sequel, by $\neg\CSP(\tempQV)$ (resp. $\CSP(\tempQV)$) we refer to the
set of all triplets $(\S,u,v)$ such that \S is a $\tau$-structure, $u,v$ are
nodes of \S and, there is no homomorphism (resp. there is a homomorphism) from
\S to $\tempQV$ sending $u$ to a source node and $v$ to a target
node\footnote{\CSP are usually defined as boolean problems, i.e. without the
  nodes $u,v$. As \RPQ queries are binary, these parameters are necessary for
  our presentation. 
  The problem $\CSP(\tempQV)$, as defined here, can be viewed as a classical $\CSP$
  problem by extending the signature with two unary predicates, interpreted as the source and 
  the target nodes, as done in~\cite{calvanese2000view}.}.
In view of Proposition~\ref{prop:cert-csp}, if
\V determines \Q in a monotone way, $(\V(\D),u,v) \in \neg\CSP(\tempQV)$ iff
$(u,v) \in \Q(\D)$.

Observe that $\neg\CSP(\tempQV)$ naturally defines a binary query associating with each 
$\tau$-structure $\S$ the set of all pairs $(u, v)$ of nodes of $\S$ such that 
$(\S,u,v) \in \neg\CSP(\tempQV)$.
By abuse of notation, when clear from the context, we will let $\neg\CSP(\tempQV)$  also denote this binary query.

\begin{rem}
  The structure $\tempQV$ of Proposition~\ref{prop:cert-csp} can be effectively computed 
  from $\Q$  and $\V$. Moreover observe that $\CSP(\tempQV)$ can be expressed in existential 
  \MSO. This shows, as mentioned in Remark~\ref{remark-mso}, that if $\V$ and $\Q$ are \RPQs and 
  $\V$ determines $\Q$ in monotone way, then there always exists a rewriting of \Q
  using \V definable in (universal) \MSO;  moreover this rewriting can be effectively computed 
  from $\Q$ and $\V$.
\end{rem}

It is well known that the certain answers query is a rewriting that can be
computed in \conp (this follows for instance from
Proposition~\ref{prop:cert-csp}). Assuming \conp is not \ptime, $\cert_{\Q,
  \V}$ cannot always be computed in polynomial time, not even under the
assumption that $\V$ determines $\Q$ in a monotone way.  Indeed it has been
shown~\cite{calvanese2000answering} that there exists $\Q$ and $\V$ defined by
\RPQ{}s such that $\cert_{\Q, \V}$ has \conp-hard data complexity. An easy
reduction from this problem shows that the lower bound remains valid if we
further assume that \V determines \Q in a monotone way:

\begin{prop}\label{prop:cert-hard}
  There exist an \RPQ view \V and an \RPQ query \Q such that \V determines \Q
  in a monotone way and it is \conp-hard to decide -- given a $\tau$-structure
  \S and nodes $(u, v)$ of \S  -- whether $(u,v) \in
  \cert_{\Q,\V}(\S)$.
\end{prop}

We show in the next section that when \V determines \Q in a monotone way there
is another rewriting expressible in \Datalog, hence computable in polynomial
time.  Before we do this we remark that the \conp complexity of $\cert_{\Q,\V}$
can be extended to Context-Free Path Query views and \RPQ queries. 

\begin{prop}
	\label{prop-certainanswersanyv}
	Let $\V$ be a Context-Free Path Query view and $\Q$ be a \RPQ. Then
	$\cert_{\Q,\V}$ can be evaluated with \conp data complexity.
\end{prop}

\begin{proof}
  Let $\V$ be a Context-Free Path Query view, and $\Q$ be a \RPQ over some
  schema $\sigma$. We prove that $\cert_{\Q,\V}$ can be evaluated with
  \conp data complexity by reducing it to the case of regular path
  views. Let $A = \langle S,\delta,q_0,F \rangle$ be a deterministic minimal
  automaton for $L(\Q)$. In what follows, $\delta(.,w)$ denotes the function
  from $S$ to $S$ associating to a state $p$ the state reached by $A$ when
  reading $w$ starting from $p$. For all $V \in \V$, we define the language $L_V$
 as: 
$$L_V = \set{w \in \sigma^* \ 
	| \ \exists w' \in L(V) ~~\delta(\cdot,w) = \delta(\cdot,w')}.$$

      We claim that $L_V$ is a regular language. To see this recall that for
      each function $f$ from $S$ to $S$ the language $L_f$ defined as
      $L_f=\set{w \in \sigma^* \ | \ \delta(\cdot,w)=f}$ is regular and notice
      that $L_V$ is a union of such languages. We remark here for later that $L_V$ is
      constructible as soon that it is decidable whether $L(V) \cap L_f$ is non
      empty. This is in particular the case when $L(V)$ is context-free.

We now define a new view $\tilde \V$ defined as the \RPQ view:
$$\tilde \V = \set{\tilde V \ | \ V\in \V \text{ and } L(\tilde V)=L_V}$$
Let $\S$ be a view instance
for $\V$. We define $\tilde \S$ as a copy of $\S$ where each $V$ relation is
replaced by $\tilde V$. Hence, $\tilde \S$ is a view instance for $\tilde
\V$. We now show that:
	$$\cert_{\Q,\V}(\S) = \cert_{\Q,\tilde \V}(\tilde \S)$$
        and thus $\cert_{\Q,\V}(\S)$ can be evaluated in \conp in the size of
        $\tilde \S$, which is also the size of $\S$.

	\begin{itemize}
    \item Assume that $(u,v) \in \cert_{\Q,\tilde \V}(\tilde \S)$. Hence,
      for all $\D$ such that $\tilde \V(\D) \supseteq \tilde \S$, there
      exists a path $\pi$ from $u$ to $v$ such that $\lambda(\pi) \in L(\Q)$. 
			Let $\D$ be a database such that $\V(\D) \supseteq \S$. Remark that, for
			all $V \in \V$, $L(V) \subseteq L(\tilde V)$. Hence, $\tilde \V(\D) 
			\supseteq \tilde \S$.  Hence, there exists
      a path $\pi$ in $\D$ from $u$ to $v$ such that $\lambda(\pi) \in
      L(\Q)$, which means that $(u,v) \in \cert_{\Q,\V}(\S)$.

    \item Conversely, assume that $(u,v) \notin 
			\cert_{\Q,\tilde \V}(\tilde \S)$. Hence, there exists a database 
			$\D$ such that $\tilde \V(\D) \supseteq \tilde \S$, but no path from $u$ 
			to $v$ in $\D$ satisfies $\Q$. From $\D$, we build a database $\D'$ as 
			follows:
      \begin{itemize}
        \item Start with $\D'$ as a copy of $\D$.
        \item For all $V \in \V$, for all $(x,y) \in \S$, if $(x,y) \in V$,
          then $(x,y) \in \tilde V$ in $\tilde \S$. We pick a path $\pi$ in
          $\D'$ from $x$ to $y$ of label $w'$ such that $w' \in
          L(\tilde V)$. Hence, there exists $w\in L(V)$ such that
          $\delta(\cdot,w') = \delta(\cdot,w)$. Then, we add in $\D'$ a
          simple path from $x$ to $y$ using only fresh nodes of label
          $w$. Hence $(x,y) \in V(\D')$.
      \end{itemize}
      Remark then that $\V(\D') \supseteq \S$. Let $\pi'$ be a path from
      $u$ to $v$ in $\D'$. Then $\pi'$ is of the form $\pi' =
      \pi_1\mu_1\pi_2\ldots\pi_{n-1}\mu_{n-1}\pi_n$, where each $\pi_i$ is
      a path that was originally in $\D$ and each $\mu_i$ is a new path
      using only fresh nodes. Then, for each $\mu_i$, there exists a path
      $\rho_i$ in $\D$ with the same starting and ending nodes and such
      that $\delta(\cdot,\lambda(\mu_i)) =
      \delta(\cdot,\lambda(\rho_i))$. Hence, we can define a path $\pi$ of
      $\D$ as $\pi =
      \pi_1\rho_1\pi_2\ldots\pi_{n-1}\rho_{n-1}\pi_n$. Hence,
      $\delta(\cdot,\lambda(\pi')) = \delta(\cdot,\lambda(\pi))$.
	
      Since $(u,v) \notin \cert_{\Q,\tilde \V}(\tilde \S)$, then
      $\delta(q_0,\lambda(\pi))\notin F$. Hence,
      $\delta(q_0,\lambda(\pi'))\notin F$, which proves that $(u,v) \notin
      \cert_{\Q,\V}(\S)$.\qedhere
	\end{itemize} 
\end{proof}

The proposition has the following consequence:

\begin{cor}\label{coro-certainanswersanyv}
	Let $\Q$ be a \RPQ and $\V$ be a Context-Free Path Query view such that $\V$
	determines $\Q$ in a monotone way. Then there exists a rewriting of $\Q$
	using $\V$ that can be evaluated with \conp data complexity.
\end{cor}

Notice that the proof of Proposition~\ref{prop-certainanswersanyv} and
therefore also Corollary~\ref{coro-certainanswersanyv} do not assume that the
language defining the views are context-free and work with any
language. However, in order to effectively construct the rewriting, it is
necessary that the formalism used to define the views has a decidable emptiness
test for the intersection with a regular language.

\subsection{Existence of a \Datalog rewriting}
We now show that for each \RPQ query \Q and each \RPQ view \V such that \V
determines \Q in a monotone way, there exists a \Datalog rewriting.

The existence of such a rewriting stems from links between \CSP{}s and
\Datalog. Recall from Proposition~\ref{prop:cert-csp} that if \V determines \Q
in a monotone way, $\neg\CSP(\tempQV)$, viewed as a binary query, is a rewriting of
 \Q using \V.  It
is known that to each \CSP problem (i.e. arbitrary template), one can associate
a canonical \Dataloglk program, for each $l, k$, with $l \leq k$. This program
can equivalently be described in terms of a two-player game, and can be thought
of as a maximal ``approximation'' of the complement of a \CSP problem, in a
precise sense (the interested reader is referred
to~\cite{feder1998computational} for more details).  Our main contribution
consists in proving that, for some explicit values of $l$ and $k$ (depending on
\Q and \V), this \Dataloglk approximation is ``exact'' when restricted to view
images (i.e. computes precisely $\neg\CSP(\tempQV)$), and is therefore a
rewriting over such instances.
 
We now present the $(l,k)$-two-player game of~\cite{feder1998computational},
and its correspondence with \Datalog. 

\begin{defi}[$(l,k)$-two-player game]
  Let $l,k$ be two integers, with $l\leq k$, let \S be a $\tau$-structure and
  $u,v$ be two nodes of \S. The $(l,k)$-game on $(\S,\tempQV,u,v)$ is played by two
  players as follows:
	\begin{itemize}
        \item The game begins with $A_0 = \emptyset$ and $h_0$ being the empty
          function over $A_0$.

\item[] For $i \geq 0$, round $i+1$ is defined as follows:
		\item Player~1 selects a set $A_{i+1}$ of nodes
                  of \S, with
			$|A_{i+1}| \leq k$ and $|A_i \cap A_{i+1}| \leq l$.
                      \item Player~2 responds by giving a homomorphism
                        $h_{i+1} : \S[A_{i+1}] \rightarrow \tempQV$ that coincides
                        with $h_i$ on $A_i \cap A_{i+1}$ and such that
                        $h_{i+1}(u)$ is a source node and $h_{i+1}(v)$ is a
                        target node whenever $u$ or $v$ are in $A_{i+1}$.
	\end{itemize}
	Player~1 wins if at any point Player~2 has no possible move. Player~2
        wins if she can play forever.
      \end{defi}

\newcommand\datuv[2]{Q_{#1,#2}}
\newcommand\datlk{\datuv{l}{k}}
\newcommand\datll{\datuv{l}{l+1}}

The existence of a winning strategy for Player~1 is expressible in \Datalog:

\begin{lem}[\cite{feder1998computational, calvanese2000view}]
	\label{dataloggame}
	Let $l,k$ be two integers, with $l \leq k$, and \Q and \V be an \RPQ
        query and an \RPQ view.  Then there exists a program $\datlk(x,y)$ in
        \Dataloglk such that for every graph database \S, $\datlk(\S)$ is the set
        of pairs $(u,v)$ such that Player~1 has a winning strategy for the
        $(l,k)$-two-player game on $(\S,\tempQV,u,v)$.
\end{lem}

Moreover the program in the above lemma can be effectively constructed from $\tempQV$, and therefore from \Q and \V. It will be simply denoted by $\datlk$ when \Q and \V are clear from the
context.

We are now ready to state the main technical result of our paper.

\begin{prop}\label{prop:datll-rewr}
Let \V and \Q be an \RPQ view and an \RPQ query such that \V determines \Q in a
  monotone way.  There exists $l$ such that $\datll$ is a rewriting of \Q using \V.
\end{prop}

Theorem~\ref{thm:datalog-rewr} is an immediate consequence of this proposition. The
rest of this section is devoted to proving Proposition~\ref{prop:datll-rewr}.
This is done in two steps. We first prove that there exists $l$ such that
$\datll$ is a rewriting of \Q using \V, when restricted to view images of
simple path graph databases. We then show that this suffices for $\datll$ to be
a rewriting of \Q using \V.

Observe that if there is a homomorphism from a $\tau$-structure $\S$ to
$\tempQV$ sending $u$ to a source node and $v$ to a target node, then Player~2
has a winning strategy for the $(l,k)$-two-player game on $(\S,\tempQV,u,v)$. This 
strategy consists in 
always playing the restriction of the homomorphism on the set selected by
Player~1.  In this sense the program $\datlk$ is a \Dataloglk under-approximation of the $\neg\CSP(\tempQV)$ problem: if $(u,v)\in\datlk(\S)$ then
$(\S,u,v) \in \neg\CSP(\tempQV)$. If moreover $\S=\V(\D)$ for some
$\sigma$-structure \D then, by Proposition~\ref{prop:cert-csp},
$(u,v)\in\datlk(\V(\D))$ implies $(u,v)\in\Q(\D)$. We will refer to this
property by saying that $\datlk$ is always \emph{sound}.

The converse  inclusion does not necessarily hold. If $(u, v) \notin
\datlk(\S)$ then Player~2 has a winning strategy, but this only means that
she can always exhibit partial homomorphisms from \S to $\tempQV$ (sometimes called
\emph{local consistency checking}); this is in general not sufficient to
guarantee the existence of a suitable global homomorphism.

However here we are not interested in arbitrary $\tau$-structures, but only
structures of the form $\V(\D)$ for some simple path graph database \D.  We now
show that, thanks to the particular properties of these structures, local
consistency checking is sufficient to obtain a global homomorphism, for some
suitable $l$ and $k=l+1$. In other words, the program $\datll$ computes
precisely $\neg\CSP(\tempQV)$ on views of simple path graph databases.

\paragraph{The case of simple path graph databases} 

\begin{prop}\label{prop-path}
  Let \V and \Q be an \RPQ view and an \RPQ query. There exists $l$ such that for every
  simple path database \D from $u$ to $v$,
\begin{equation*}
  (u,v) \in \datll(\V(\D)) \text{ iff } (\V(\D),u,v) \in \neg\CSP(\tempQV).
\end{equation*} 
In particular if \V determines \Q in a monotone way,
\begin{equation*}
(u,v) \in \datll(\V(\D)) \text{ iff } (u,v) \in \Q(\D).
\end{equation*}
\end{prop}

\begin{proof}
Let \V and \Q be an \RPQ view and an \RPQ query, and let $\D$ be a graph database 
consisting of a simple path from node $u$ to node $v$. Assume $u, v \in \V(\D)$.

We will show, in Lemma~\ref{lemma-game} below, that for large enough $l$, if
Player~2 has a winning strategy on the game on $(\V(\D),\tempQV, u,v)$ then we
can exhibit a homomorphism witnessing the fact that $(\V(\D), u, v) \in
\CSP(\tempQV)$. Before that we prove crucial properties of  
$\V(\D)$ which will be exploited in the sequel. For that 
we need the following simple definitions and claims.

Let $\D$ consist of the simple path $\pi=x_0a_1x_1\ldots x_{m-1}a_mx_{m}$,
with $x_0=u$ and $x_{m}=v$. Moreover let
 $\S=\V(\D)$ and let
 $A=\langle S_\V,\delta_\V,q^0_\V,F_\V \rangle$ be the product automaton of
all the deterministic minimal automata of all the regular expressions of the
\RPQ{}s in \V. Let $N(\V)$ be the number of states of $A$, i.e. $|S_\V|$.

In what follows, for $q\in S_\V$ and $w \in \sigma^*$, $\delta_\V(q,w)$
denotes the state $p \in S_\V$ such that there is a run of $A$ on $w$
starting in state $q$ and arriving in state $p$.

For every $k \leq m+1$, and every $i,j \leq k$,  we say that $x_i \sim_k x_j$ in $\V(\D)$ 
if, for all $V \in \V$, for all $r \geq k$,
\begin{equation*}
(x_i,x_r) \in V(\D) \quad \Leftrightarrow \quad (x_j,x_r) \in V(\D)
\end{equation*}

For all $k$, the relation $\sim_k$ is an equivalence relation over $\{x_i \ | \ i \leq k
\}$. We now prove the main property of $\V(\D)$, namely that the index of all
$\sim_k$ is bounded by the size of $\V$.

\begin{claim}\label{maxeq}
For all $k \leq m+1$:
	$$\Big |\{x_i \ | \ i \leq k \} / \sim_k \Big | 
	\leq N(\V)$$
\end{claim}
\begin{proof}
  To each node $x_i$ in $\pi$ with $i \leq k$, we associate a state
  $\varphi(x_i) \in S_\V$ defined as : $$\varphi(x_i) =
  \delta_\V(q_\V^0,\lambda(\pi_{i \rightarrow k}))$$ where
  $\pi_{s \rightarrow t}$ is defined as the subpath of $\pi$ that starts at
  position $s$ and ends at position $t$, that is $\pi_{s \rightarrow t} = x_{s}
  a_{s} x_{s+1} a_{s+1}\ldots a_{t-1} x_t$.

  Assume that there exist two nodes $x_i$ and $x_j$, with $i,j \leq k$, that
  have the same image in $\varphi$. It follows that:
	$$\delta_\V(q_\V^0,\lambda(\pi_{i \rightarrow k})) = 
	\delta_\V(q_\V^0,\lambda(\pi_{j \rightarrow k})) $$ Let us prove that
        $x_i \sim_k x_j$. Assume that there exist $r \geq k$ and $V\in
        \V$ such that $(x_i,x_r)\in V(\D)$.  Then
        $\delta_\V(q_\V^0,\lambda(\pi_{i \rightarrow r}))$ is final for $V$.
        Remark that $\lambda(\pi_{i \rightarrow r}) = \lambda(\pi_{i
          \rightarrow k})\lambda(\pi_{k \rightarrow r}$), from which we can
        deduce that :
	$$\delta_\V(q_\V^0,\lambda(\pi_{i \rightarrow r})) =	
	\delta_\V(\varphi(x_i),\lambda(\pi_{k \rightarrow r})) $$
	Hence,
	$$\delta_\V(q_\V^0,\lambda(\pi_{i \rightarrow r})) =	
	\delta_\V(\varphi(x_j),\lambda(\pi_{k \rightarrow r})) $$ We can now
        conclude that $\delta_\V(q_\V^0,\lambda(\pi_{j \rightarrow r}))$ is
        final for $V$, which means that $(x_j,x_r) \in V(\D)$. A symmetric argument  
        easily proves the other direction of the equivalence. Hence,
        $x_i \sim_k x_j$, and we can finally conclude that there cannot be more
        that $N(\V)$ distinct equivalence classes of $\sim_k$ over the
        nodes $\{x_i \ | \ i \leq k \}$ of $\pi$.
\end{proof}

The following easily verified property of the equivalence relations $\sim_k$ will 
also be useful:

\begin{claim}\label{transeq}
  Let $k_1,k_2 \leq m+1$, with $k_1 \leq k_2$. Let $x$ and $y$ be two elements
  of $\pi$ that occur before $x_{k_1}$.  Then $x \sim_{k_1} y$ implies $x
  \sim_{k_2} y$.
\end{claim}

We are now ready to prove the statement of the Proposition. 
Let $l = |\tempQV| \cdot N(\V)$.
We prove that $(u,v) \in \datll(\S)$ iff 
$(\S, u, v) \in \neg\CSP(\tempQV)$.
In view of the fact that $\datll$ encodes the $(l, l+1)$-two-player game in the sense of  Lemma~\ref{dataloggame}, it is enough to prove the following:

\begin{lem}\label{lemma-game}
  {\sloppy Player~2 has a winning strategy for the $\,(l,l+1)\,$-\,two-player game on
  $\,(\S,\tempQV,u,v)$ iff there is an homomorphism from \S to
  $\tempQV$ sending $u$ to a source node and $v$ to a target node.}
\end{lem}
\begin{proof}
The right-left direction is obvious. If there is a suitable homomorphism $h : \S
\rightarrow \tempQV$, then Player~2 has a winning strategy which consists in
playing according to $h$.

Conversely, assume that Player~2 has a winning strategy for the
$(l,l+1)$-two-player game on $(\S,\tempQV, u, v)$. Let $\{s_1,s_2,\ldots,s_r\}$ be an
ordering of the elements of \S, according to the order on $\pi$, that is, in
such a way that $\forall j\leq k$, $s_j$ occurs before $s_k$ in $\pi$. 
Clearly $s_1 =u$ and $s_r=v$. If $r\leq l +1$, Player~1 can select all elements of \S in a single round, and then
Player~2 has to provide a full homomorphism from \S to $\tempQV$,
which concludes the proof.

Assume $r > l+1$. For ease of notations, we will number rounds starting from
$l+1$. This can be seen just as a technicality, or equivalently as Player~1
selecting the empty set for the first $l$ rounds.  Since Player~2 has a winning
strategy, she has, in particular, a winning response against the following play
of Player~1 :
	\begin{itemize}
        \item On round $l+1$, Player~1 plays $A_{l+1} =
          \{s_1,\ldots,s_{l+1}\}$. Player~2 has to respond with a partial
          homomorphism $h_{l+1}$, which she can do, since she has a winning
          strategy.
        \item Assume that, on round $i$, $A_i$ is of size $l+1$ and its element
          of biggest index is $s_{i}$  (as it is the case on round $l+1$). 
          Given the choice of $l$, the set $A_i$ is sufficiently ``big", that is
          by Claim~\ref{maxeq}, there exist two elements $s_j,s_k \in
          A_i$ such that $s_j \sim_{i} s_k$, and $h_i(s_j) = h_i(s_k)$. On
          round $i+1$, Player~1 picks $A_{i+1} = (A_{i} - \{s_j\}) \cup
          \{s_{i+1}\}$. This choice maintains that $A_{i+1}$ is of size
          $l+1$ and that its element of
          biggest index is $s_{i+1}$. Once again, Player~2 has to respond with
          a partial homomorphism $h_{i+1}$, which she can do.
		\item Following this play, on round $r$, $A_{r}$ contains $s_r$, the element of biggest index in \S.
			From now on, we no longer care about Player~1's move, that is, we arbitrarily set $A_i = \emptyset$ for all 
			$i > r$.
	\end{itemize}
	We can now define $h$ as follows :
	$$h(s_i) =
		\left\{
		\begin{array}{ll}
			h_{l+1}(s_i)  & \mbox{if } i \leq l+1 \\
			h_{i}(s_i) & \mbox{if } l+1 < i \leq r
		\end{array}
		\right.
	$$
	Observe that, by definition, the mapping $h$ sends $u$ to a source node and $v$ to a target node
	(since so do all the $h_i$'s used in the game). It remains to prove that 
	$h$ is an homomorphism from \S to $\tempQV$. We prove by
	induction on $i \geq l+1$ that : 
	\begin{enumerate}[label=$(H_{\arabic*})$]
        \item[$(H_1)$] $h$ is a homomorphism from $\S[\{s_1,\ldots,s_i\}]$
          to $\tempQV$.
        \item[$(H_2)$] $h$ coincides with $h_i$ on $A_i$.
        \item[$(H_3)$] for all $j \leq i$, there exists $s \in A_i$ such that
          $s_j \sim_i s$ and $h(s_j) = h(s)$.
	\end{enumerate}

	\textbf{Base case :} For $i = l+1$, the mapping $h$ coincides by definition with
        $h_{l+1}$ on $\{s_1,\ldots,s_{l+1}\}$. Hence, $(H_1)$ and $(H_3)$
        follow easily.

	\textbf{Inductive case :} Assume that there exists $i$ with $l+1 \leq i < r$
        such that $(H_1)$,$(H_2)$ and $(H_3)$ holds for $i$; we prove them for 
        $i+1$.

	\begin{enumerate}[label=$(H_{\arabic*})$]
        \item[$(H_2)$]Let $s \in A_{i+1}$. If $s = s_{i+1}$, then, by
          definition, $h(s_{i+1}) = h_{i+1}(s_{i+1})$. Otherwise, $s\in A_i
          \cap A_{i+1}$. $(H_2)$ for $i$ implies that $h(s) = h_i(s)$, and the
          definition of $h_{i+1}$ thus yields $h_{i+1}(s) = h_i(s) =
          h(s)$. Hence, $(H_2)$ holds for $i+1$.

        \item[$(H_3)$]Let $j \leq i+1$. If $j = i+1$, then $s_j \in A_{i+1}$,
          and the result is obvious.  Otherwise, $(H_3)$ for $i$ implies that
          there exists $s \in A_i$ such that $s_j \sim_i s$ and $h(s_j) =
          h(s)$.  From Claim~\ref{transeq}, we deduce that $s_j \sim_{i+1}
          s$. If $s\in A_{i+1}$, there is nothing more to prove.  Otherwise, it
          means that $s$ is exactly the element that was removed from $A_i$ on
          round $i+1$, which means that there exists another element $s' \in
          A_{i}\cap A_{i+1}$ such that $s \sim_{i} s'$ and $h_i(s) = h_i(s')$.
          Then Claim~\ref{transeq} and $(H_2)$ imply that $s_j \sim_{i+1} s'$
          and $h(s_j) = h(s')$. Hence $(H_3)$ holds for $i+1$.

        \item[$(H_1)$]By definition, $h$ already preserves any self-loop. 
          Moreover, $(H_1)$ for $i$ implies that $h$
          is a homomorphism from $\S[\{s_1,\ldots,s_i\}]$ to
          $\tempQV$. Hence, any edge between two elements of
          $\{s_1,\ldots,s_i\}$ in $\S$ is already preserved by $h$. Let $s_j \in
          \{s_1,\ldots,s_i\}$.  Remark that, since $\pi$ is a simple path,
          there are no edges from 
          $s_{i+1}$ to $s_j$ in $\S$. Thus, we just have to prove that all edges from
          $s_j$ to $s_{i+1}$ are preserved by $h$.

          $(H_3)$ for $i+1$ implies that there exists an element $s\in A_{i+1}$
          such that $s_j \sim_{i+1} s$ and $h(s_j) = h(s)$. Since $h_{i+1}$ is
          a homomorphism on $\S[A_{i+1}]$, it preserves all edges
          from $s$ to $s_{i+1}$. Moreover, $(H_2)$ for $i+1$ implies that $h$
          and $h_{i+1}$ coincide on $A_{i+1}$, which means that $h$ preserves
          all edges from $s$ to $s_{i+1}$. Finally, the definition of
          $\sim_{i+1}$ implies that $s_j$ and $s$ have the same edges to
          $s_{i+1}$. Hence, $h$ preserves all edges from $s_j$ to $s_{i+1}$.
        \end{enumerate}
	Finally, $(H_1)$ applied for $r$ proves that $h$ is indeed a
        homomorphism from \S to $\tempQV$.
        This completes the proof of Lemma~\ref{lemma-game}.
      \end{proof} 
      Now assume  \V determines \Q in a monotone way, then from Proposition~\ref{prop:cert-csp} 
      it immediately follows that $(u,v) \in \datll(\V(\D))$ iff $(u,v) \in \Q(\D)$.
      This completes the proof of Proposition~\ref{prop-path}.
\end{proof}
    
    \paragraph{From simple paths  to arbitrary graph databases}

    Proposition~\ref{prop-path} shows that if \Q determines \V in a monotone
    way then $\datll$ is a rewriting of \Q using \V, when restricted to simple path
    databases. It remains to lift this result to arbitrary graph databases.  In
    a sense, the following result shows that the general case can always be
    reduced to the simple path case.

\begin{prop}\label{prop-general-to-path}
  Let \V and \Q be an \RPQ view and an \RPQ query such that \V determines \Q in a
  monotone way.  Assume $\mathcal{P}$ is a query of schema $\tau$ such that:

	\begin{enumerate}
        \item\label{item-monotone} $\mathcal{P}$ is closed under homomorphisms: for all
          databases $\S,\S'$, and all pair of elements $(u,v)$ of
          \S, if $(u,v) \in \mathcal{P}(\S)$ and there exists a homomorphism
          $h:\S \rightarrow \S'$ then $(h(u),h(v)) \in \mathcal{P}(\S')$.
			
        \item\label{item-complete}  $\mathcal{P}$ is sound and complete for all simple path
          databases: 
          for all simple path databases $\D$ from $u$ to $v$  such that $u$ and $v$ are in the    
          domain of $\V(\D)$, we have $(u,v) \in \mathcal{P}(\V(\D))$ iff
          $(u,v)\in \Q(\D)$.
			
        \item\label{item-sound} $\mathcal{P}$ is always sound: for all graph databases \D and
          elements $u$ and $v$ of $\V(\D)$, if $(u,v) \in \mathcal{P}(\V(\D))$
          then $(u,v) \in \Q(\D)$.
	\end{enumerate}
	Then $\mathcal{P}$ is a rewriting of \Q using \V.
      \end{prop}
      
\begin{proof}
  Let \D be a database, and $(u,v)$ be a pair of elements of $\V(\D)$,
  such that $(u,v) \in \Q(\D)$.  Then there exists in \D a path $\pi_0$ from $u$
  to $v$, such that $\lambda(\pi_0) \in L(\Q)$.
	
  Consider the simple path $\pi = x_0 a_0 x_1 \ldots x_m a_m x_{m+1}$ defined
  such that $\lambda(\pi) = \lambda(\pi_0)$.  Since $\V$ determines $\Q$ in a
  monotone way and $\lambda(\pi) \in L(\Q)$, then $x_0$ and $x_{m+1}$ are in
  the domain of $\V(\pi)$, and $(x_0,x_{m+1})\in \Q(\pi)$.
  Hence,~\eqref{item-complete} implies that $(x_0,x_{m+1})\in
  \mathcal{P}(\V(\pi))$.
	
  Additionally, it is clear that there exists a homomorphism $h$ from $\pi$ to
  \D with $h(x_0) = u$ and $h(x_{m+1}) = v$. Observe that $h$ extends to the views of
  $\pi$ and $\D$, that is $h$ is an homomorphism from $\V(\pi)$ to $\V(\D)$,
  and~\eqref{item-monotone} thus implies that $(u,v)\in \mathcal{P}(\V(\D))$.
	
	The other direction is immediately given by~\eqref{item-sound}.
\end{proof}

\noindent We now have all the elements to prove Proposition~\ref{prop:datll-rewr}.
Let \V and \Q be an \RPQ view and an \RPQ query such that \V determines \Q in a
monotone way.  
By Proposition~\ref{prop-path} there exists $l$ such that $\datll$ is sound and complete  
over simple path databases. 
Moreover each \Datalog  query is preserved under homomorphisms, and we have already observed that all $\datlk$ are always sound.
It then follows from Proposition~\ref{prop-general-to-path} that there exists $l$ such that $\datll$ is a rewriting of \Q using \V.
This proves Proposition~\ref{prop:datll-rewr} and therefore Theorem~\ref{thm:datalog-rewr}.
 
\section{Conclusions}

We have seen that if an \RPQ view \V determines an \RPQ query \Q in a monotone way
then a \Datalog rewriting can be computed from \V and \Q.
As a corollary it is decidable whether there exists a \Datalog rewriting to
an \RPQ query using \RPQ views.

These results extends to 2-way-\RPQ. A 2-way-\RPQ is defined using a regular
expression over the alphabet $\sigma\cup\bar\sigma$. It asks for pairs of nodes
linked by a 2-way-path using the symbol $a$ for traversing an edge of label $a$
in the direction of the arrow, and the symbol $\bar a$ for backward traversing an edge of
label $a$. This query language has been
studied in~\cite{calvanese07tcs}. In particular~\cite{calvanese07tcs} gives an extension of
Corollary~\ref{cor:mon-det} and of Proposition~\ref{prop:cert-csp} for
2-way-\RPQ. Building from these two results it is possible to extend the
results of Section~\ref{section-datalog} to 2-way-\RPQ{}s. The details are more
complicated and omitted here, but the general idea is the same.

We may wonder whether a simpler query language than \Datalog could suffice to express monotone rewritings of \RPQ queries using \RPQ views.
For instance all examples we are aware of use only the transitive closure of binary
Conjunctive Regular Path Queries. It is then natural to ask whether linear \Datalog (where at most
one internal predicate may occur in the body of each rule), using internal
predicates of arity at most 2, can express all monotone rewritings. We leave
this interesting question for future work.

Finally we conclude by mentioning that we don't know yet whether the monotone
determinacy problem for Conjunctive Regular Path Query is decidable. Likewise,
deciding whether an \RPQ view determines an \RPQ query, without the
monotonicity assumption, is still an open problem.

\bibliographystyle{plain}
\bibliography{biblio}

\end{document}